\documentclass[pdflatex,sn-mathphys-num]{sn-jnl}

\usepackage{bm}
\usepackage{graphicx}
\usepackage{amssymb}
\usepackage{amsmath, amscd}  
\usepackage{amsthm}
\newtheorem{theorem}{Theorem}

\newcommand{\R}{\mathbb{R}}

\newcommand{\ket}[1]{|#1\rangle}

\makeatletter
\@ifundefined{orcidlogo}
  {\newcommand{\orcidlogo}{}}
  {\renewcommand{\orcidlogo}{}}
\makeatother

\begin{document}

\title{Random-matrix reduction in projective quantum mechanics}

\author{\fnm{Alexey A.} \sur{Kryukov}}\email{kryukov@uwm.edu}

\affil{\orgdiv{Department of Mathematics \& Natural Sciences}, \orgname{University of Wisconsin-Milwaukee}, \orgaddress{\street{2200 E. Kenwood Blvd.}, \city{Milwaukee}, \postcode{WI 53211}, \country{USA}}}

\abstract{%
We develop a state-space geometric framework for measurement, classicality, and quantum paradoxes, based on one dynamical conjecture. Classical configuration space and classical phase space for a mechanical system arise as distinguished submanifolds of projective quantum state space. On these submanifolds, the Fubini--Study geometry induces Euclidean classical geometry, and the tangent component of Schr\"odinger evolution reproduces Newtonian dynamics. Within this framework, interactions with measuring devices and environments are described by random-matrix dynamics on projective state space, generated by matrices drawn from the Gaussian Unitary Ensemble. We show that this random-matrix dynamics yields isotropic diffusion, giving Born-rule transition probabilities in microscopic measurements and stabilizing classical behavior in macroscopic systems. We further argue that the random-matrix conjecture is not an independent ad hoc assumption: under natural translation-invariance assumptions on the distribution of state-space steps originating on the classical submanifold, the unitary lift of homogeneous and isotropic Brownian motion on that submanifold is uniquely given by the Gaussian Unitary Ensemble, up to scale and an irrelevant scalar part. The resulting framework provides a unitary account of measurement and the quantum-to-classical transition and, if accepted, offers a dynamical resolution of standard quantum paradoxes.%
}


\maketitle

\section{Introduction}

Quantum mechanics contains several closely related foundational problems. The Schr\"odinger equation is linear and unitary, yet measurements produce individual outcomes rather than their superpositions. The probabilities of these outcomes are given by the Born rule, but this rule does not follow directly from Schr\"odinger dynamics. Macroscopic bodies follow Newtonian trajectories, although their states must also be governed by quantum dynamics. Classical space appears as the arena of classical dynamics and of all observed events, but its relation to quantum state space remains unclear. These problems are usually discussed under separate names: the measurement problem, the origin of the Born rule, the preferred-basis problem, the quantum-to-classical transition, and the emergence of classical space.

The purpose of this paper is to show that these problems have a common origin and admit a common experimentally testable resolution based on a geometric framework and a single dynamical conjecture. The geometric framework follows from the Schr\"odinger dynamics and from the fact that classical space and classical phase space can be represented by distinguished submanifolds of projective quantum state space. These submanifolds have Euclidean induced geometry and support Newtonian dynamics as the tangent component of Schr\"odinger evolution. In this sense, classical mechanics is not external to quantum mechanics; it is contained in quantum mechanics as tangent dynamics on a submanifold of projective state space.

The resulting geometry provides the setting for the dynamical conjecture. We assume that, under coarse-graining, the effect of many short and complicated interactions of a particle, or a system of particles, with a measuring device or environment is described by a random-matrix Hamiltonian. More specifically, the effective interaction Hamiltonian is modeled by independent draws from the Gaussian Unitary Ensemble. Under this random-matrix conjecture, denoted below by {\bf (RM)}, the state evolution becomes an isotropic random walk on projective state space.
%
%The conjecture and its consequences are analogous to Einstein's derivation of Brownian motion, a process that may be used to describe measurement errors in classical physics. There, too, one does not begin with a complete microscopic description of all individual interactions, but with effective assumptions of small independent increments, homogeneity, and isotropy. Additional motivation for {\bf (RM)} comes from the well-known ubiquity of random matrices in the description of fluctuations in quantum systems.

The conjecture and its consequences are closely analogous to Einstein's derivation of Brownian motion, a process that may be used to describe measurement errors in classical physics. There, too, one does not begin with a complete microscopic description of all individual interactions, but with effective assumptions of small independent increments, homogeneity, and isotropy. In the present framework, Brownian motion on the classical-space submanifold is the restriction of the corresponding state-space random walk. Conversely, under the assumption that the distribution of state-space steps originating on this submanifold is invariant under spatial translations, the unitary lift of a homogeneous and isotropic Gaussian random walk on the classical-space submanifold is generated by random Hermitian matrices from the Gaussian Unitary Ensemble, up to scale and an irrelevant scalar part. Thus {\bf (RM)} is not merely analogous to Brownian motion; under these conditions, it is the unique unitary lift of Brownian measurement dynamics. Additional motivation for {\bf (RM)} comes from the well-known ubiquity of random matrices in the description of fluctuations in quantum systems.

The geometric framework and the conjecture {\bf (RM)} have complementary roles. 
The geometry explains where classical variables come from and why Schr\"odinger dynamics has a Newtonian component. 
The random-matrix conjecture explains why states of macroscopic systems remain close to the classical sector and why microscopic measurements yield individual outcomes with Born probabilities. 
When the {\bf (RM)}-induced diffusion is restricted to the classical space submanifold, it becomes ordinary Brownian motion and gives the normal distribution of classical measurement errors. 
In the full projective state space, the same diffusion gives the Born rule. 
Thus the normal distribution in classical measurement and the Born rule in quantum measurement are two manifestations of a single state-space diffusion process.

Mathematically, the classical space and phase-space submanifolds of state space are formed from equivalence classes of states whose position standard deviation is bounded above by a parameter \(\sigma\). Physically, \(\sigma\) is related to the resolution of position-measuring devices. A measuring device distinguishes only equivalence classes of states that differ at the given resolution. Thus a classical outcome of a position measurement is not a single wave function, but an equivalence class of sufficiently localized states. This point is essential for explaining why a state undergoing diffusion in projective state space has a nonzero probability of reaching the finite-resolution classical sector, despite the finite dimensionality of the corresponding classical submanifold. It also avoids the usual problem of ``tails.''

The resulting picture is as follows. Before measurement, a microscopic system evolves according to the usual Schr\"odinger dynamics. During the measurement interaction, the state may evolve through the full projective state space, where the alternatives associated with different classes can interfere. The \({\bf (RM)}\)-induced isotropic diffusion then leads to a record in one of the detector-defined equivalence classes. The probability of recording a given class is given by the Born rule, in the same sense that a classical diffusion gives the probability of finding a particle in a specified region.

For a macroscopic body, frequent environmental interactions continually return the state to a narrow neighborhood of the classical phase-space submanifold. Between such interactions, the tangent component of the Schr\"odinger flow gives Newtonian drift. The resulting stochastic process is a random walk in projective state space with intermittent conditioning on returns to detector-defined localized sectors. For macroscopic bodies, this produces a sequence of localized records whose conditional distributions remain narrow around the classical path. The observed motion is therefore a stroboscopic Newtonian trajectory.

It is important that the primary dynamics is not motion of particles in \(\mathbb R^3\), but motion of states in projective state space. The framework therefore should not be interpreted as a hidden trajectory theory in ordinary space, or as an attempt to visualize quantum motion as particle motion in \(\mathbb R^3\). In particular, unlike Bohmian mechanics, it does not postulate definite particle positions guided by the wave function. Rather, the state itself evolves in projective state space, which becomes the primary arena for the physical processes described here. Classical motion in \(\mathbb R^3\) appears only when the state is located near the classical phase-space submanifold of state space.

Assuming that {\bf (RM)} holds as an effective coarse-grained description, and using the geometric representation of classical space and phase space as submanifolds of state space, one obtains strict unitary evolution at the fundamental level, individual outcomes with Born probabilities, and recovery of Newtonian motion within a single framework. These consequences are mathematically precise and physically nontrivial. They also place the standard quantum paradoxes in a new light: many of them arise from assigning classical properties to states that do not lie on the appropriate classical submanifold, or from identifying measurement outcomes with exact rays rather than finite-resolution equivalence classes. The aim of the paper is to make this framework explicit and to show that the parameter regimes required for macroscopic classicality are physically plausible.

The paper is organized as follows. The first part develops the geometric framework and reviews the embedding of classical configuration space and phase space into projective quantum state space. In particular, it shows how the classical action and Newtonian equations arise from the Schr\"odinger action by restriction to the classical phase-space submanifold. The next part introduces the random-matrix conjecture and explains its relation to Brownian motion, isotropic diffusion, equivalence classes, and the Born rule. This is followed by a formulation of macroscopic classicality as a conditioned stochastic process combining Newtonian tangent drift and {\bf (RM)} diffusion. The subsequent estimates show that physically reasonable parameter choices yield Newtonian behavior for macroscopic bodies and Born-rule reduction for microscopic measurements. The paper then discusses the double-slit experiment, cloud-chamber tracks, Stern--Gerlach measurement, and macroscopic superpositions in this framework. It concludes with a comparison to decoherence, continuous measurement, collapse models, and other approaches to the quantum-to-classical transition.

Some of the material presented here appeared previously in
\cite{KryukovPhysicsA,KryukovPHLA} and in earlier publications by the author. The present work,
however, gives a more complete and unified treatment. It includes new results,
theorems, and physically motivated estimates, as well as a detailed discussion
of standard quantum-mechanical paradoxes within this framework and a comparison
with existing approaches. A separate companion paper \cite{KryukovComp} provides
extensive numerical simulations supporting the main theoretical results.

\section{Classical submanifolds of projective state space}
\label{sec:classical-submanifolds}

Let the one-particle Hilbert space be \(L_2(\mathbb R^3)\), and let \(\mathbb{CP}^{L_2}\) denote the corresponding projective state space with the Fubini--Study metric. For a small localization parameter \(\sigma>0\), consider states represented by functions of the form
\begin{equation}
\label{localized-packet}
\varphi_{{\bf a},{\bf p}}({\bf x})
=
r_{{\bf a},\sigma}({\bf x})\,e^{i{\bf p}\cdot{\bf x}/\hbar},
\end{equation}
where
\begin{equation}
\label{localized-profile}
r_{{\bf a},\sigma}({\bf x})
=
\sigma^{-3/2}
r\!\left(\frac{{\bf x}-{\bf a}}{\sigma}\right).
\end{equation}
Here \(r\in L_2(\mathbb R^3)\) is a fixed normalized real-valued sufficiently regular function, centered at the origin, with finite variance normalized to \(1\). As \(\sigma\to0\), the densities \(r_{{\bf a},\sigma}^2({\bf x})\) converge in the sense of distributions to \(\delta^3({\bf x}-{\bf a})\) \cite{delta}.

Let \(M_{3,3}^{\sigma}\) be the image of the map
\begin{equation}
\label{Omega-map}
\Omega:\mathbb R^3\times\mathbb R^3\longrightarrow \mathbb{CP}^{L_2},
\qquad
({\bf a},{\bf p})\longmapsto \varphi_{{\bf a},{\bf p}}.
\end{equation}
Similarly, let \(M_{3}^{\sigma}\) be the image of the corresponding restricted map
\begin{equation}
\label{omega-map}
\omega:\mathbb R^3\longrightarrow \mathbb{CP}^{L_2},
\qquad
{\bf a}\longmapsto r_{{\bf a},\sigma}.
\end{equation}

\begin{theorem}%[Classical phase space as a state-space submanifold]
\label{thm:phase-submanifold}
The images \(M_{3,3}^{\sigma}\) and \(M_{3}^{\sigma}\) of the maps \(\Omega\) and \(\omega\) are embedded submanifolds of \(\mathbb{CP}^{L_2}\). The metrics induced on them by the ambient Fubini--Study metric, equivalently the pull-backs of the Fubini--Study metric under \(\Omega\) and \(\omega\), are Euclidean after an appropriate rescaling of units. Thus \(M_{3,3}^{\sigma}\) and \(M_{3}^{\sigma}\) are isometric to \(\mathbb R^3\times\mathbb R^3\) and \(\mathbb R^3\), respectively.
\end{theorem}

\begin{proof}[Sketch of proof]
The maps \(\Omega\) and \(\omega\) are differentiable embeddings, and their images are the submanifolds \(M_{3,3}^{\sigma}\) and \(M_{3}^{\sigma}\) of \(\mathbb{CP}^{L_2}\). The restrictions of these maps provide coordinate parametrizations of the corresponding submanifolds. The tangent space to \(M_{3,3}^{\sigma}\) is spanned by the derivatives
\[
\frac{\partial\varphi_{{\bf a},{\bf p}}}{\partial a^i},
\qquad
\frac{\partial\varphi_{{\bf a},{\bf p}}}{\partial p^j}.
\]
The Fubini--Study metric is obtained from the Hilbert inner product after quotienting by the vertical phase direction. Pulling this metric back by \(\Omega\), direct computation shows that the \({\bf a}\)-directions and \({\bf p}\)-directions are orthogonal and have constant norms. Hence, after rescaling the coordinates, the induced metric on \(M_{3,3}^{\sigma}\) is Euclidean. Accordingly, \(M_{3,3}^{\sigma}\), equipped with the metric induced from the ambient Fubini--Study metric, is isometric to \(\mathbb R^3\times\mathbb R^3\) in these coordinates. The same pull-back computation for \(\omega\) gives the corresponding Euclidean metric on \(M_3^\sigma\), making \(M_3^\sigma\) isometric to \(\mathbb R^3\).
\end{proof}

We now show that the submanifold \(M_{3,3}^\sigma\) also carries the classical dynamics. Consider the action functional
\begin{equation}
\label{Schrodinger-action}
S[\varphi]
=
\int
\overline{\varphi}({\bf x},t)
\left[
i\hbar\frac{\partial}{\partial t}
-
\widehat h
\right]
\varphi({\bf x},t)\,d^3{\bf x}\,dt,
\end{equation}
where
\begin{equation}
\label{Hamiltonian}
\widehat h
=
-\frac{\hbar^2}{2m}\Delta
+
\widehat V({\bf x},t).
\end{equation}
Variation of \eqref{Schrodinger-action} with respect to \(\varphi\) gives the Schr\"odinger equation.

\begin{theorem}%[Classical action from the restricted Schr\"odinger action]
\label{thm:restricted-action}
The restriction of the Schr\"odinger action \eqref{Schrodinger-action} to \(M_{3,3}^{\sigma}\) gives the classical action, up to a total derivative, an additive \(\sigma\)-dependent constant independent of \(({\bf a},{\bf p})\), and a potential error that vanishes as \(\sigma\to0\), assuming continuity of \(V\).
\end{theorem}

\begin{proof}[Sketch of proof]
We restrict \(\varphi\) in \eqref{Schrodinger-action} to \(M_{3,3}^{\sigma}\) by taking
\[
\varphi({\bf x},t)
=
r_{{\bf a}(t),\sigma}({\bf x})
e^{i{\bf p}(t)\cdot{\bf x}/\hbar}.
\]
For this restricted ansatz, the time-derivative term satisfies
\[
i\hbar\langle \varphi,\partial_t\varphi\rangle
=
{\bf p}\cdot\dot{\bf a}
+
\frac{d}{dt}(\text{total derivative term}),
\]
where \(\langle \cdot, \cdot \rangle\) denotes the inner product and the contribution from the real amplitude vanishes by normalization. The kinetic term satisfies
\[
\left\langle
\varphi,
-\frac{\hbar^2}{2m}\Delta\varphi
\right\rangle
=
\frac{{\bf p}^2}{2m}
+
C_\sigma,
\]
where
\[
C_\sigma
=
\frac{\hbar^2}{2m\sigma^2}
\int_{\mathbb R^3}|\nabla r({\bf y})|^2\,d^3{\bf y}
\]
is independent of \({\bf a}\) and \({\bf p}\). Thus the non-potential terms already have the classical form, modulo inessential total-derivative and additive constant terms.

The only genuine approximation comes from the potential term
\[
\int_{\mathbb R^3}
V({\bf x},t)r_{{\bf a},\sigma}^2({\bf x})\,d^3{\bf x}.
\]
Since \(r_{{\bf a},\sigma}^2\rightharpoonup \delta^3({\bf x}-{\bf a})\) as \(\sigma\to0\), this term converges to \(V({\bf a},t)\) whenever \(V\) is continuous at \({\bf a}\). Therefore, assuming continuity of the potential and taking \(\sigma\) sufficiently small, the restricted action takes the classical form
\begin{equation}
\label{classical-action}
S
=
\int
\left[
{\bf p}\cdot\frac{d{\bf a}}{dt}
-
h({\bf p},{\bf a},t)
\right]dt,
\end{equation}
where
\begin{equation}
\label{classical-Hamiltonian}
h({\bf p},{\bf a},t)
=
\frac{{\bf p}^2}{2m}
+
V({\bf a},t).
\end{equation}
\end{proof}

If \(V\) is differentiable, variation of \eqref{classical-action} gives Hamilton's equations
\begin{equation}
\label{Hamilton-equations}
\frac{d{\bf a}}{dt}
=
\frac{{\bf p}}{m},
\qquad
\frac{d{\bf p}}{dt}
=
-\nabla V({\bf a},t),
\end{equation}
which are the Newton equations. In this sense, differentiability of \(V\) is needed not for the convergence of the action itself, but for deriving the Newtonian equations of motion from the resulting classical action.

The same result can be formulated directly in terms of the Schr\"odinger vector field on projective state space. Let
\[
X_{\widehat h}(\varphi)
=
-\frac{i}{\hbar}\widehat h\varphi
\]
denote the Schr\"odinger velocity. For \(\varphi\in M_{3,3}^{\sigma}\), decompose \(X_{\widehat h}\) into a component tangent to \(M_{3,3}^{\sigma}\) and an orthogonal component:
\[
X_{\widehat h}
=
X_T+X_\perp.
\]

\begin{theorem}%[Newtonian tangent component]
\label{thm:tangent-newton}
The tangent component \(X_T\) of the Schr\"odinger vector field on \(M_{3,3}^{\sigma}\) is the Newtonian vector field on classical phase space. In coordinates \(({\bf a},{\bf p})\), it is given by
\[
X_T
=
\dot{\bf a}\cdot\partial_{\bf a}
+
\dot{\bf p}\cdot\partial_{\bf p},
\]
where \(\dot{\bf a}=d{\bf a}/dt\) and \(\dot{\bf p}=d{\bf p}/dt\) satisfy \eqref{Hamilton-equations}.
\end{theorem}

\begin{proof}[Sketch of proof]
The tangent vectors \(\partial_{\bf a}\varphi\) and \(\partial_{\bf p}\varphi\) span \(T_\varphi M_{3,3}^{\sigma}\). Projecting the Schr\"odinger velocity onto this tangent space gives precisely the coefficients obtained from the constrained variation of the restricted action. Therefore the tangent component of the Schr\"odinger flow is the Hamiltonian flow of \eqref{classical-Hamiltonian}, i.e. Newtonian motion.
\end{proof}

The orthogonal component \(X_\perp\) measures the tendency of the state to leave the classical submanifold. For Gaussian representatives with variance \(\sigma^2\), the norm of the orthogonal component is
\[
\|X_\perp\|
=
\frac{\hbar}{4\sqrt{2}\,m\sigma^2},
\]
which is the usual spreading scale of a localized packet. For sufficiently massive bodies, this normal component is small, although its effect can accumulate over time. Thus the geometric construction explains the form of the classical equations, but by itself it does not explain why macroscopic states remain close to the classical submanifold. 

The preceding results justify the identification of \({\bf a}\) and \({\bf p}\) with the classical position and momentum variables. The parameter \({\bf a}\) labels the center of localization of the packet, while \({\bf p}\) enters through the linear phase factor \(e^{i{\bf p}\cdot{\bf x}/\hbar}\). The localized state is not a momentum eigenstate; rather, \({\bf p}\) determines the local wave vector and hence the group velocity of the packet. More generally, for a sufficiently narrow packet with smooth phase \(\Theta({\bf x})\), expansion near the center \({\bf a}\) gives
\[
\Theta({\bf x})
=
\Theta({\bf a})
+
\nabla \Theta({\bf a})\cdot({\bf x}-{\bf a})
+
O(|{\bf x}-{\bf a}|^2).
\]
The constant term is projectively irrelevant, while the linear term is represented by \(e^{i{\bf p}\cdot{\bf x}/\hbar}\). Thus a general localized state has a linear phase to leading order, with corrections controlled by the localization scale \(\sigma\). This justifies the choice of \(\varphi_{{\bf a},{\bf p}}({\bf x})\) in (\ref{localized-packet}) in the regime of sufficiently small \(\sigma\).

It follows that  \(M_{3,3}^{\sigma}\), equipped with the induced Fubini--Study metric and the tangent Schr\"odinger dynamics, is identified with the classical phase-space submanifold, while \(M_3^\sigma\) is its configuration-space part. The normal component of the Schr\"odinger vector field measures the tendency of the state to leave this classical sector.
The free parameter \(\sigma\) controls the localization scale of the submanifolds \(M_3^\sigma\) and \(M_{3,3}^\sigma\). As discussed in the Introduction, we interpret \(\sigma\) physically as the resolution scale of position-measuring devices. This motivates an equivalence-class formulation, in which localized states indistinguishable at resolution \(\sigma\) are identified. The equivalence-class formulation introduced next gives the classical sector its operational meaning in terms of finite-resolution measurements.
%%%

\subsection{Equivalence-class version of the classical submanifolds}
\label{equiv}

For simplicity, consider first one spatial dimension. Let
\[
g_{c,\sigma}(z)
=
\left(\frac{1}{2\pi\sigma^2}\right)^{1/4}
\exp\!\left[-\frac{(z-c)^2}{4\sigma^2}\right]
\]
denote the normalized Gaussian of width \(\sigma\) centered at \(c\). 
In one dimension, the map
\[
\omega:\mathbb R\to \mathbb{CP}^{L_2},
\qquad
\omega(c)=g_{c,\sigma},
\]
parametrizes a one-dimensional submanifold of \(\mathbb{CP}^{L_2}\), denoted here by \(M_1^\sigma\).

For Gaussian representatives in one dimension, one has
\begin{equation}
\label{M1-distance}
\cos^2\rho(g_{c,\sigma},g_{d,\sigma})
=
\exp\!\left[-\frac{(c-d)^2}{4\sigma^2}\right].
\end{equation}
It follows that
\[
\rho(g_{c,\sigma},g_{d,\sigma})
=
\frac{|c-d|}{2\sigma}
+
o\!\left(\frac{|c-d|}{2\sigma}\right),
\]
and hence \(M_1^\sigma\), equipped with the induced Fubini--Study metric, is isometric to \(\mathbb R\) with its Euclidean distance rescaled by \(1/(2\sigma)\).

Let \(\mu_z(\psi)\) and \(\delta_z(\psi)\) denote the position expectation value and position standard deviation of a normalized state \(\psi\in L_2(\mathbb R)\). 
For a fixed resolution \(\sigma\), define the equivalence class associated with the classical position \(c\) by
\begin{equation}
\label{position-equivalence-class}
\{g_c\}
=
\left\{
\psi\in L_2(\mathbb R):
\mu_z(\psi)=c,\quad \delta_z(\psi)\le \sigma
\right\}.
\end{equation}
The Gaussian state \(g_{c,\sigma}\) is a convenient representative of the equivalence class. Each class contains infinitely many mutually orthogonal states. All states within a given class are experimentally indistinguishable by a position-measuring device whose resolution is no finer than \(\sigma\).

The Fubini--Study distance from a state \(\psi\) to an equivalence class is defined by
\begin{equation}
\label{distance-to-class}
\rho(\psi,\{g_c\})
=
\inf_{\chi\in\{g_c\}}\rho(\psi,\chi),
\end{equation}
and the distance between two equivalence classes is defined by
\begin{equation}
\label{distance-between-classes}
\rho(\{g_c\},\{g_d\})
=
\inf_{\psi\in\{g_c\}}\rho(\psi,\{g_d\}).
\end{equation}
For nearby classes, the infimum in (\ref{distance-between-classes}) has the same leading-order asymptotics as the distance between the Gaussian representatives \(g_{c,\sigma}\) and \(g_{d,\sigma}\). Thus
\begin{equation}
\label{class-distance-gaussian}
\rho(\{g_c\},\{g_d\})
=
\frac{|c-d|}{2\sigma}
+
o\!\left(\frac{|c-d|}{2\sigma}\right),
\end{equation}
which is the same relation obtained for the representative manifold \(M_1^\sigma\). Therefore the set of equivalence classes
\[
\widetilde M_1^\sigma
=
\{\{g_c\}:c\in\mathbb R\}
\]
inherits the same Euclidean metric as \(M_1^\sigma\) and is isometric to \(\mathbb R\).

The corresponding phase-space representative manifold is obtained by adjoining
the linear phase factor \(e^{ipz/\hbar}\) to the Gaussian representatives. Namely,
define
\[
M_{1,1}^\sigma
=
\left\{
g_{c,\sigma}e^{ipz/\hbar}: c,p\in\mathbb R
\right\}
\subset \mathbb{CP}^{L_2}.
\]
For two representatives in \(M_{1,1}^\sigma\), direct computation gives
\begin{equation}
\label{M11-distance}
\cos^2\rho
\left(
g_{a,\sigma}e^{ipz/\hbar},
g_{b,\sigma}e^{iqz/\hbar}
\right)
=
\exp\!\left[
-\frac{(a-b)^2}{4\sigma^2}
-
\frac{\sigma^2(p-q)^2}{\hbar^2}
\right].
\end{equation}
Consequently, for nearby points,
\[
\rho^2
\left(
g_{a,\sigma}e^{ipz/\hbar},
g_{b,\sigma}e^{iqz/\hbar}
\right)
=
\frac{(a-b)^2}{4\sigma^2}
+
\frac{\sigma^2(p-q)^2}{\hbar^2}
+
o\!\left(
\frac{(a-b)^2}{4\sigma^2}
+
\frac{\sigma^2(p-q)^2}{\hbar^2}
\right).
\]
Thus the induced Fubini--Study metric on \(M_{1,1}^\sigma\) is
\[
ds^2
=
\frac{dc^2}{4\sigma^2}
+
\frac{\sigma^2}{\hbar^2}\,dp^2,
\]
where \(c\) and \(p\) are the position and momentum coordinates, respectively.
After rescaling the position and momentum coordinates, \(M_{1,1}^\sigma\) is
isometric to the Euclidean phase plane \(\mathbb R\times\mathbb R\).

The equivalence-class phase-space version is obtained in the same way, by
adjoining the phase factor \(e^{ipz/\hbar}\) to the equivalence classes of
real-valued localized states. Namely, define
\[
\widetilde M_{1,1}^\sigma
=
\left\{
\{g_c\}e^{ipz/\hbar}: c,p\in\mathbb R
\right\},
\]
where \(\{g_c\}\) consists of real-valued states with position expectation \(c\)
and position standard deviation not exceeding \(\sigma\). The distance from a
state to an equivalence class is defined as above, and the metric on
\(\widetilde M_{1,1}^\sigma\) is defined using the Gaussian representatives
\(g_{c,\sigma}e^{ipz/\hbar}\). Therefore the same formula
\eqref{M11-distance} determines the induced metric on
\(\widetilde M_{1,1}^\sigma\). Hence \(\widetilde M_{1,1}^\sigma\) inherits,
after the same rescaling of the position and momentum coordinates, the Euclidean
metric on the classical phase plane \(\mathbb R\times\mathbb R\).

To avoid confusion, we note that \(\widetilde M_1^\sigma\) and
\(\widetilde M_{1,1}^\sigma\) differ by the additional phase, or momentum,
label. The set of localized states underlying \(\widetilde M_1^\sigma\)
contains the set underlying \(\widetilde M_{1,1}^\sigma\), since multiplying a
localized state by \(e^{ipz/\hbar}\) does not change its position expectation
or position spread. Conversely, for a sufficiently localized state with smooth
phase, the phase is linear to leading order near the center of localization, so
it is represented, up to corrections controlled by \(\sigma\), by such a factor.
Thus, for the purposes of localization and position recording, the two
underlying localized sectors may be regarded as the same, while the
equivalence-class manifold \(\widetilde M_{1,1}^\sigma\) retains the additional
momentum information.

The construction extends componentwise to three dimensions. The equivalence-class manifolds
\[
\widetilde M_3^\sigma,
\qquad
\widetilde M_{3,3}^\sigma
\]
are obtained by grouping states with the same position expectation value \({\bf a}\) and position spread bounded by \(\sigma\), with the phase factor \(e^{i{\bf p}\cdot{\bf x}/\hbar}\) added in the phase-space case. They inherit the same induced Euclidean geometry as the representative manifolds \(M_3^\sigma\) and \(M_{3,3}^\sigma\).

The reduction of Schr\"odinger dynamics to Newtonian dynamics also passes to these equivalence-class manifolds. Indeed, the representative calculation above shows that the tangent component of the Schr\"odinger vector field on \(M_{3,3}^\sigma\) is Newtonian. Since all states in a class are indistinguishable at resolution \(\sigma\), and since the classical variables \(({\bf a},{\bf p})\) label the classes, the induced tangent dynamics on \(\widetilde M_{3,3}^\sigma\) is the same Newtonian dynamics in the variables \(({\bf a},{\bf p})\), up to errors below the chosen resolution. Thus the passage from representatives to equivalence classes does not change the classical equations; it gives them their operational meaning.

\subsection{Many-particle systems}
\label{many-p}

The same construction applies to systems of many particles. For \(N\) distinguishable particles with Hilbert space
\[
\mathcal H_N
=
L_2(\mathbb R^3)^{\otimes N}
\simeq
L_2(\mathbb R^{3N}),
\]
the classical configuration-space sector is represented by tensor products of localized position classes,
\begin{equation}
\label{N-particle-position-sector}
\widetilde M_{3,N}^{\sigma}
=
\widetilde M_3^{\sigma_1}\otimes\cdots\otimes
\widetilde M_3^{\sigma_N}.
\end{equation}
A point of this manifold is labeled by
\[
({\bf a}_1,\ldots,{\bf a}_N),
\]
where \({\bf a}_k\in\mathbb R^3\) is the position coordinate of the \(k\)-th particle. Similarly, the classical phase-space sector is
\begin{equation}
\label{N-particle-phase-sector}
\widetilde M_{3,3,N}^{\sigma}
=
\widetilde M_{3,3}^{\sigma_1}\otimes\cdots\otimes
\widetilde M_{3,3}^{\sigma_N},
\end{equation}
with coordinates
\[
({\bf a}_1,\ldots,{\bf a}_N;
{\bf p}_1,\ldots,{\bf p}_N).
\]

For tensor-product representatives, the Hilbert-space overlap factorizes into the product of the one-particle overlaps. Since the Fubini--Study distance satisfies
\[
\cos^2\rho(\psi,\phi)=|\langle \psi,\phi\rangle|^2,
\]
the distance $\rho$ between representatives is determined by this product. In particular, for localized Gaussian representatives one obtains
\begin{equation}
\label{N-particle-distance}
\cos^2\rho
=
\exp\!\left[
-\sum_{k=1}^{N}
\frac{({\bf a}_k-{\bf b}_k)^2}{4\sigma_k^2}
-
\sum_{k=1}^{N}
\frac{\sigma_k^2({\bf p}_k-{\bf q}_k)^2}{\hbar^2}
\right].
\end{equation}
Thus, after introducing the corresponding dimensionless coordinates, the induced metric on \(\widetilde M_{3,3,N}^{\sigma}\) is Euclidean. Consequently,
\[
\widetilde M_{3,N}^{\sigma}
\cong
\mathbb R^{3N},
\qquad
\widetilde M_{3,3,N}^{\sigma}
\cong
\mathbb R^{3N}\times\mathbb R^{3N},
\]
as Riemannian manifolds with the induced metric.

The usual classical picture of \(N\) particles in a single copy of physical space \(\mathbb R^3\) is recovered through the canonical identification
\[
({\bf a}_1,\ldots,{\bf a}_N)\in\mathbb R^{3N}
\quad
\longleftrightarrow
\quad
\text{\(N\) points }{\bf a}_1,\ldots,{\bf a}_N
\text{ in one space } \mathbb R^3 .
\]
Thus the tensor product of \(N\) localized one-particle classical sectors is the usual \(N\)-particle configuration space, equivalently the space of \(N\)-point configurations in one classical space.

The dynamical statement also extends. If the \(N\)-particle Schr\"odinger Hamiltonian has the form
\begin{equation}
\label{N-particle-Hamiltonian}
\widehat h
=
-\sum_{k=1}^{N}
\frac{\hbar^2}{2m_k}\Delta_k
+
\widehat V({\bf x}_1,\ldots,{\bf x}_N,t),
\end{equation}
then restriction of the Schr\"odinger action to
\(\widetilde M_{3,3,N}^{\sigma}\) gives the classical action
\begin{equation}
\label{N-particle-classical-action}
S
=
\int
\left[
\sum_{k=1}^{N}
{\bf p}_k\cdot \dot{\bf a}_k
-
h({\bf a}_1,\ldots,{\bf a}_N;
{\bf p}_1,\ldots,{\bf p}_N;t)
\right]dt,
\end{equation}
where
\begin{equation}
\label{N-particle-classical-Hamiltonian}
h
=
\sum_{k=1}^{N}
\frac{{\bf p}_k^2}{2m_k}
+
V({\bf a}_1,\ldots,{\bf a}_N,t).
\end{equation}
If \(V\) is differentiable, constrained variation yields the Newton equations
\begin{equation}
\label{N-particle-Newton}
\dot{\bf a}_k=\frac{{\bf p}_k}{m_k},
\qquad
\dot{\bf p}_k
=
-\nabla_{{\bf a}_k}V({\bf a}_1,\ldots,{\bf a}_N,t),
\qquad
k=1,\ldots,N.
\end{equation}
Accordingly, the many-particle classical world appears as the tangent Schr\"odinger dynamics on the corresponding many-particle localized sector of projective state space.

The results of this section are mathematical statements about localized submanifolds of projective quantum state space. 
They show that classical configuration space and classical phase space, for one particle and for many-particle systems, can be represented inside projective Hilbert space with the correct Euclidean geometry and Newtonian tangent dynamics.
These results do not by themselves force an ontological identification. One could regard \(M_3^\sigma\), \(M_{3,3}^\sigma\), and their equivalence-class versions merely as embedded copies of classical spaces inside \(\mathbb{CP}^{L_2}\). However, these submanifolds have precisely the structures expected of classical space and phase space: Euclidean induced geometry, classical variables \(({\bf a},{\bf p})\), and Newtonian tangent dynamics. One can, therefore, assume that the classical world is not an additional arena external to quantum mechanics. It is the distinguished localized sector of projective state space.

This assumption is stronger than what is needed merely to derive Newtonian equations from the restricted Schr\"odinger action or to obtain the later results of the paper. Its role is conceptual and structural. It allows classical records, measurement outcomes, and macroscopic configurations to be described within the same state space in which microscopic quantum superpositions evolve. This identification will be used below when the random-matrix dynamics is added: microscopic measurement, macroscopic classicality, and the standard quantum paradoxes will then be treated as different regimes of one state-space dynamics. However, all of these results can be derived without the identification and we leave it up to the reader to either use the unified state space framework, or to keep the classical picture separate.

%%%

\section{Randomness, equivalence classes, and the random-matrix conjecture}
\label{sec:rm}

The preceding section identifies the classical sector of quantum theory: classical configuration space and phase space appear as localized submanifolds of projective state space, and the tangent component of the Schr\"odinger flow on these submanifolds reproduces Newtonian dynamics. To address why a microscopic measurement yields a single outcome and why the state of a macroscopic body remains close to the classical submanifold rather than wandering into the full projective space, one must add a dynamical mechanism describing the effect of measuring devices and environments.

Brownian motion provides the guiding analogy for introducing such a dynamical mechanism. In classical physics, a measurement of position does not produce an exact point; it produces a recorded value with finite resolution and random error. When the measured particle or the pointer of a measuring device is subject to many small uncontrolled interactions with its environment, these errors may be modeled by Brownian motion. The observed result is then distributed around the classical value, typically according to a normal distribution determined by the diffusion process and the observation period. Thus Brownian motion gives a dynamical model of classical measurement uncertainty.

In Einstein's theory of Brownian motion \cite{Ein}, one does not begin with a detailed account of molecular collisions. Instead, one assumes that the particle undergoes many small, approximately independent random displacements, with a distribution of steps that is homogeneous in space and time and isotropic in the absence of drift. These assumptions are sufficient to derive the diffusion equation and, for a localized initial condition, the normal distribution. The conjecture introduced below is the corresponding state-space analogue: during measurement or environmental interaction, the quantum state undergoes many small random unitary changes, which after coarse-graining produce an isotropic random walk in projective Hilbert space.

\subsection{The random-matrix conjecture}

Let \(\varphi\in \mathbb{CP}^{L_2}\) be the state of the system. During a short interaction interval \(dt\), we write the effective interaction Hamiltonian as a random Hermitian operator \(\widehat h_{\rm RM}\). The corresponding infinitesimal evolution is
\begin{equation}
\label{RM-step}
\varphi \longmapsto
\exp\!\left(-\frac{i}{\hbar}\widehat h_{\rm RM}dt\right)\varphi .
\end{equation}
The conjecture is the following.

\medskip

\noindent
\textbf{(RM)}
\emph{During measurement or environmental monitoring, the effective interaction Hamiltonians acting on the relevant finite-dimensional coarse-grained subspace are modeled, over successive short time intervals, by independent random matrices drawn from the Gaussian Unitary Ensemble.}

\medskip

The finite-dimensional restriction in this statement reflects the finite resolution of any actual measurement. A detector distinguishes only a finite number of alternatives at a given resolution and within a finite observation time. The infinite-dimensional projective space is therefore approached through finite-dimensional coarse-grained sectors, with the continuum limit taken only after the induced geometric behavior is identified.

The appearance of the Gaussian Unitary Ensemble in \({\bf (RM)}\) is not arbitrary. It can be motivated from the classical Brownian limit itself. The translated Gaussian states forming \(M_3^\sigma\) constitute a complete family in the Hilbert state space $L_2(\R^3)$: if a state is orthogonal to all translated Gaussians, then it is zero in $L_2(\R^3)$. Equivalently, this follows from the standard Fourier-transform argument, since the Fourier transform of a Gaussian is nowhere zero. The corresponding family of tangent vectors, obtained by differentiating the translated Gaussians, is also complete: a state orthogonal to all such tangent vectors would have Fourier transform supported only at \(k=0\), hence would be constant, and therefore zero in \(L_2(\mathbb R^3)\). Thus, although each tangent space \(T_\psi M_3^\sigma\) is finite-dimensional, the family of tangent spaces obtained by translating \(\psi\) through \(M_3^\sigma\) probes the full state space.

The following result expresses the sense in which \({\bf (RM)}\) is the unitary lift of classical Brownian measurement dynamics. The reason for looking for a unitary lift is the relation between Newtonian and Schr\"odinger dynamics established above. Newtonian dynamics on the classical phase-space submanifold is obtained by restricting, or projecting, the Schr\"odinger dynamics to that submanifold. Since Schr\"odinger evolution is unitary, any state-space process whose restriction gives a stochastic Newtonian process should likewise be represented by a stochastic unitary evolution. Brownian motion, understood as a coarse-grained stochastic process arising from Newtonian dynamics, should therefore appear as the restriction of a unitary stochastic process on projective state space.

\begin{theorem}%[Gaussian unitary lift of Brownian motion]
\label{thm:Brownian-lift-RM}
Suppose that Brownian motion on \(M_3^\sigma\) is represented by a homogeneous and isotropic Gaussian random walk. Suppose also that its lift to projective state space is required to be a Schr\"odinger evolution, and hence unitary, and that the distribution of state-space steps originating at points of \(M_3^\sigma\) is invariant under spatial translations \({\bf a}\rightarrow{\bf a}+\Delta{\bf a}\). Then the Hamiltonian generating the lifted walk is distributed according to the Gaussian Unitary Ensemble, with scale fixed by the Brownian diffusion coefficient on \(M_3^\sigma\). Scalar multiples of the identity contribute only projectively irrelevant phases.
\end{theorem}

\begin{proof}[Sketch of proof]
By the assumed unitary lift, an infinitesimal step has the form
\[
\psi\mapsto e^{-i\widehat h\,dt/\hbar}\psi ,
\]
where \(\widehat h=\widehat h^\ast\). 
For \(\psi\in M_3^\sigma\), the tangential component of the infinitesimal step is
\[
\operatorname{proj}_{T_\psi M_3^\sigma}
\left(
-\frac{i}{\hbar}\widehat h\psi\,dt
\right).
\]
By assumption, these tangential increments reproduce the homogeneous and isotropic Gaussian random walk on \(M_3^\sigma\). Translation invariance of the state-space step distribution means that the same law is sampled at all translated localized states, corresponding to spatial translations
\[
{\bf a}\mapsto {\bf a}+\Delta{\bf a}.
\]
Since the translated tangent directions of \(M_3^\sigma\) form a complete family in the Hilbert state space, the covariance of the lifted generator is fixed by its action on a complete family of directions. Isotropy of the Brownian motion on \(M_3^\sigma\) forces the covariance of the induced tangent increments to be a scalar multiple of the identity on each \(T_\psi M_3^\sigma\). By completeness, the covariance of the lifted generator is likewise fixed as a scalar multiple of the identity on the corresponding state-space directions.

Thus \(\widehat h\) is a centered Gaussian Hermitian operator with invariant covariance. In a finite-dimensional coarse-grained sector, the corresponding law is the centered Gaussian law on Hermitian matrices invariant under unitary conjugation,
\[
\widehat h\mapsto U\widehat h U^{-1}.
\]
This is precisely the Gaussian Unitary Ensemble. Its overall scale is fixed by matching the tangential covariance of the induced walk on \(M_3^\sigma\) with the Brownian diffusion coefficient. The scalar part of \(\widehat h\), although present in the ensemble, contributes only a global phase and therefore does not affect the induced motion in projective state space.
\end{proof}

The preceding argument can be simplified if one assumes from the outset that the distribution of state-space steps is homogeneous and isotropic with respect to the Fubini--Study geometry. In that case, the covariance of the infinitesimal displacement is already the same at every point and in every tangent direction of projective state space. Therefore, once the scale is fixed by matching the induced Brownian covariance in a single tangent direction at one point of \(M_3^\sigma\), the entire Gaussian unitary ensemble is determined. The translation-invariance and completeness argument above is needed only to obtain this state-space isotropy from the weaker and more physically transparent assumptions of spatial translation invariance along \(M_3^\sigma\) and isotropy of the Brownian motion on \(M_3^\sigma\).

A simple finite-dimensional analogy helps clarify the uniqueness of the lift. Consider Brownian motion of a bead constrained to a helix in \(\mathbb R^3\). The helix is one-dimensional, but its tangent directions form a complete set in \(\mathbb R^3\): no nonzero vector in \(\mathbb R^3\) is orthogonal to all of them. If an ambient Gaussian random walk in \(\mathbb R^3\) has a translation-invariant distribution of steps and induces the given homogeneous Brownian motion along the helix, then the tangential data determine the ambient covariance, and hence the ambient Gaussian walk, up to its diffusion scale. In the present setting, \(M_3^\sigma\) plays the role of the helix, projective Hilbert space plays the role of the ambient space, and unitary Schr\"odinger evolution replaces arbitrary Euclidean steps by Hermitian generators.

%The preceding theorem gives the Brownian-to-random-matrix direction. Conversely, once the Gaussian Unitary Ensemble is used in \({\bf (RM)}\), the resulting state-space walk is automatically homogeneous and isotropic. Its restriction to the classical-space submanifold gives the corresponding Brownian motion.

It is useful to distinguish two closely related local descriptions of a random walk on a submanifold. A constrained random walk on a submanifold is defined intrinsically: its small steps lie in the tangent directions of the submanifold, and the resulting path remains on the submanifold. Equivalently, it may be modeled by allowing only those infinitesimal ambient steps that are tangent to the submanifold. A projected random walk begins with steps in the ambient space and then, for each step, retains its tangent component. For homogeneous and isotropic Gaussian small steps, these two descriptions have the same local tangential distribution after the diffusion scale is matched. Thus, at the infinitesimal level relevant for Brownian motion, the constrained walk on the submanifold agrees with the tangent projection of an ambient isotropic Gaussian walk.

In the present setting, Brownian motion on \(M_3^\sigma\) is the constrained classical walk, while \({\bf (RM)}\) gives the ambient unitary walk in projective state space. The preceding theorem showed that, under the stated translation-invariance and completeness assumptions, Brownian motion on \(M_3^\sigma\) has \({\bf (RM)}\) as its Gaussian unitary lift. Conversely, an \({\bf (RM)}\)-induced state-space walk gives Brownian motion on \(M_3^\sigma\) when its infinitesimal displacements are restricted to their tangential components along \(M_3^\sigma\).

\begin{theorem}%[Brownian motion induced on \(M_3^\sigma\)]
\label{thm:RM-Brownian-on-M}
For initial states \(\psi\in M_3^\sigma\), the tangential component of the \({\bf (RM)}\)-generated state-space displacement defines a homogeneous and isotropic Gaussian random walk on \(M_3^\sigma\). In the small-step limit, this walk converges to Brownian motion on \(M_3^\sigma\), with diffusion coefficient fixed by the scale of the GUE distribution.
\end{theorem}

\begin{proof}[Sketch of proof]
Let \(\psi\in M_3^\sigma\). Under an infinitesimal \({\bf (RM)}\) step,
\[
\psi\mapsto e^{-i\widehat h_{\rm RM}dt/\hbar}\psi ,
\]
the first-order displacement in projective state space is represented by
\[
-\frac{i}{\hbar}\widehat h_{\rm RM}\psi\,dt ,
\]
after quotienting by the vertical phase direction. Its tangential component along
\(M_3^\sigma\) is
\[
\operatorname{proj}_{T_\psi M_3^\sigma}
\left(
-\frac{i}{\hbar}\widehat h_{\rm RM}\psi\,dt
\right).
\]

Since \(\widehat h_{\rm RM}\) is drawn from the Gaussian Unitary Ensemble, this
tangential displacement is Gaussian. The unitary invariance of the GUE implies
that the distribution of the ambient infinitesimal displacement is the same at
all points of projective state space and has no preferred projective direction.
Consequently, its tangent projection to \(T_\psi M_3^\sigma\) has the same
covariance at all points \(\psi\in M_3^\sigma\). Moreover, the induced covariance
on each \(T_\psi M_3^\sigma\) is a scalar multiple of the identity with respect
to the induced Fubini--Study metric.

Here \(M_3^\sigma\) is a real, indeed totally real, submanifold of complex projective state space. The induced increments on \(T_\psi M_3^\sigma\) therefore have ordinary real orthogonal isotropy with respect to the induced Fubini--Study metric. At the same time, these increments arise as tangent components of an ambient unitary random walk in projective state space. The ambient walk is generated by Hermitian operators drawn from the GUE, whose unitary invariance gives transitivity on complex projective state space. Thus the real orthogonal isotropy observed on \(M_3^\sigma\) is the restriction of the unitary-invariant state-space dynamics to the real classical submanifold.

Thus the tangent projection of the \({\bf (RM)}\)-induced state-space walk has the local covariance of Brownian motion on \(M_3^\sigma\). Passing to the small-step limit gives Brownian motion on \(M_3^\sigma\), with diffusion coefficient fixed by the scale of the GUE distribution.
\end{proof}

The following property of \({\bf (RM)}\)-induced motion is central to the derivation of the Born rule and to the explanation of classical stability.

\begin{theorem}%[Isotropy of RM-induced motion]
\label{thm:RM-isotropy}
The random unitary steps generated by \({\bf (RM)}\) induce a homogeneous and isotropic random walk on projective Hilbert space with respect to the Fubini--Study metric.
\end{theorem}

\begin{proof}[Sketch of proof]
In the finite-dimensional sector in which \({\bf (RM)}\) is applied, the GUE measure on Hermitian matrices is invariant under
\[
\widehat h_{\rm RM}\mapsto U\widehat h_{\rm RM}U^{-1}
\]
for every unitary \(U\). Therefore the infinitesimal unitary transformation generated by \(\widehat h_{\rm RM}\) has a distribution invariant under the induced action of the unitary group on projective space. Since this action preserves the Fubini--Study metric and is transitive, the displacement distribution is the same at every point of projective space. Moreover, the stabilizer of each point acts transitively on directions of equal Fubini--Study norm in the tangent space, so the displacement law is isotropic in each tangent space. This is precisely homogeneity and isotropy of the induced state-space random walk.
\end{proof}

The Gaussian Unitary Ensemble is therefore singled out in three complementary ways. First, using the completeness of the translated Gaussian states and their tangent vectors, a homogeneous and isotropic Brownian walk on \(M_3^\sigma\) has a unique Gaussian unitary lift under spatial translation invariance. This lift is generated by GUE Hamiltonians, with the overall scale fixed by the Brownian diffusion coefficient. Second, once GUE is used as the distribution of random Hamiltonians, the induced tangential walk on \(M_3^\sigma\) is Brownian motion on the classical-space submanifold. Third, GUE gives a homogeneous and isotropic random walk in projective Hilbert space. Note that the same conclusion would not hold for the Gaussian Orthogonal Ensemble, because the orthogonal group does not act transitively on complex projective space.

The conjecture {\bf (RM)} is not a microscopic derivation of the interaction Hamiltonian from first principles. However, the preceding results show that its form is strongly constrained, and in this sense nearly forced, by the geometric framework. Once Brownian measurement dynamics on the classical-space submanifold is required to admit a translation-invariant unitary lift to projective state space, the effective Hamiltonian must be drawn from the Gaussian Unitary Ensemble, with scale fixed by the Brownian diffusion coefficient. Thus {\bf (RM)} is not an arbitrary stochastic postulate, but the natural state-space analogue of the assumptions used in Einstein's derivation of Brownian motion. This conclusion is also consistent with the familiar role of random matrices in modeling universal fluctuation statistics in complex quantum systems \cite{Wigner, BGS}. Its content is that the coarse-grained effect of many short, complicated, and uncontrollable interaction events is an isotropic random walk in projective state space, whose induced tangential motion on the classical-space submanifold reproduces Brownian measurement dynamics. The physical adequacy of this effective description must ultimately be tested by its consequences and by the plausibility of the resulting parameter values. These estimates will be discussed below.

\subsection{Equivalence classes and measurement outcomes}

A measuring device cannot distinguish all points of projective Hilbert space. In the present framework,
a recorded outcome is therefore not identified with an individual state vector,
but with one of the equivalence classes of states indistinguishable by the
device at the given resolution. As in Section \ref{equiv}, for a position measurement in one spatial direction,
let
\[
\mu_z(\varphi)=\langle \varphi,\widehat z\varphi\rangle
\]
denote the expectation value of position, and let
\[
\delta_z(\varphi)
=
\left(
\langle \varphi,(\widehat z-\mu_z)^2\varphi\rangle
\right)^{1/2}
\]
denote the corresponding position standard deviation. A finite-resolution
position outcome is represented by a class of states with the same value of
\(\mu_z\), interpreted as the recorded position at the given resolution, and
with position uncertainty bounded by the resolution scale \(\sigma\), up to an
inessential numerical factor. Thus a classical position outcome is represented
not by a single wave function, but by an equivalence class localized within the
detector resolution.

For the reduction process it is useful to separate the position expectation value and
the degree of localization from the remaining directions in state space. Define
\begin{equation}
\label{F-map}
F(\varphi)=(\mu_z,\delta_z).
\end{equation}
The level sets
\[
\mu_z=\tau,
\qquad
\delta_z=\lambda
\]
form a codimension-two foliation of the relevant regular part of state space.
Translating and scaling a suitable initial state produces a two-dimensional
surface \(M_\varphi\subset\mathbb{CP}^{L_2}\), with coordinates
\[
\tau=\mu_z,
\qquad
s=\ln(\delta_z/\sigma),
\]
where the same resolution scale \(\sigma\) is used as reference. Translation
changes the mean position \(\tau\), while scaling changes the width
\(\delta_z\), equivalently the coordinate \(s\).

To see the geometry of these coordinates explicitly, one may work in a
finite-dimensional approximation of state space generated by a finite
superposition of mutually orthogonal, or approximately orthogonal, localized
functions. For instance, one may use sufficiently narrow Gaussian states with
well-separated centers. In such a representation, the \(\tau\)- and
\(s\)-coordinates on \(M_\varphi\) are orthogonal with respect to the induced
Fubini--Study metric \cite{KryukovPHLA,KryukovPhysicsA}.
The finite-resolution localized sector associated with resolution \(\sigma\)
is described by
\[
s\le 0.
\]
This is the sector in which the state is sufficiently localized to be assigned
to one of the classical position-outcome classes introduced above.

When the walk in \({\bf (RM)}\) is expressed in the \((\tau,s)\)-coordinates on
\(M_\varphi\), the \(\tau\)-component describes motion along the recorded-position
direction, while the \(s\)-component describes motion toward or away from the
localized sector. Since these coordinates are orthogonal and the \({\bf (RM)}\)
step distribution is isotropic, the corresponding infinitesimal increments in
\(\tau\) and \(s\) are independent Gaussian variables. Motion of the state along
the leaves of the foliation does not change \(\mu_z\) or \(\delta_z\), and
therefore does not contribute to state reduction.

The state diffuses in projective state space, and a recorded outcome occurs
when the path reaches the localized sector \(s\le0\), corresponding to one of
the detector-defined equivalence classes. 
Since the \({\bf (RM)}\) distribution
is isotropic, the induced walk in \(s\) is symmetric and has no drift.
Hence, after
sufficiently many steps, the probability that the state is found in the
localized sector at a given observation time is approximately \(1/2\). As we now show,
conditional on reaching the localized sector, the relative probability of
reaching a particular outcome sector satisfies the Born rule.

\subsection{Normal distribution and Born rule}

The same {\bf (RM)}-induced isotropic random walk has two different manifestations, depending on whether it is constrained to a classical submanifold or allowed to evolve in the full projective state space. When constrained to the classical submanifold \(M_3^\sigma\simeq \mathbb R^3\), the Fubini--Study metric reduces to the Euclidean metric. The {\bf (RM)} process then becomes a random walk in \(\mathbb R^3\) with independent, isotropic Gaussian increments. In the small-increment limit, this process converges to Brownian motion, and the resulting transition probabilities are normal distributions centered at the classical position.
In the full projective state space, the same {\bf (RM)} process is homogeneous and isotropic with respect to the Fubini--Study metric. Hence transition probabilities can depend only on projective distance. 

For a single normalized final state \(\psi\), the Born rule may be written as
\[
P(\varphi\to\psi)=|\langle \psi,\varphi\rangle|^2
=
\cos^2\rho(\varphi,\psi),
\]
where \(\rho(\varphi,\psi)\) is the Fubini--Study distance. For an outcome represented by a closed subspace \(\mathcal H_\alpha\), or by the corresponding projective sector, the probability is
\[
P(\alpha)=\|P_\alpha\varphi\|^2,
\]
where \(P_\alpha\) is the orthogonal projection onto \(\mathcal H_\alpha\). In particular, for a position measurement in a region \(\Delta\subset\mathbb R^3\),
\[
P(\Delta)=\|P_\Delta\varphi\|^2
=
\int_\Delta |\varphi({\bf x})|^2\,d^3{\bf x},
\]
where \(P_\Delta\varphi({\bf x})=\chi_\Delta({\bf x})\varphi({\bf x})\).

This projection rule can be expressed geometrically as a distance to the outcome sector. Let \(\mathbb P(\mathcal H_\alpha)\) denote the projectivization of \(\mathcal H_\alpha\), and define
\[
\rho(\varphi,\mathbb P(\mathcal H_\alpha))
=
\inf_{\psi\in\mathcal H_\alpha,\ \|\psi\|=1}
\rho(\varphi,\psi).
\]
Since \(\rho(\varphi,\psi)=\arccos |\langle \psi,\varphi\rangle|\), one obtains
\[
\cos^2\rho(\varphi,\mathbb P(\mathcal H_\alpha))
=
\sup_{\psi\in\mathcal H_\alpha,\ \|\psi\|=1}
|\langle \psi,\varphi\rangle|^2
=
\|P_\alpha\varphi\|^2.
\]
Thus the Born probability of an outcome sector is the squared cosine of the Fubini--Study distance from the initial state to that sector. For a position region \(\Delta\), this gives
\[
P(\Delta)
=
\cos^2\rho(\varphi,\mathbb P(\mathcal H_\Delta))
=
\int_\Delta |\varphi({\bf x})|^2\,d^3{\bf x}.
\]

The use of sectors is essential. A physical measurement does not require the random path to hit an exact ray or an exact lower-dimensional submanifold. The detector-defined outcome is an equivalence class, or a finite-resolution neighborhood of such a class. The probability of an outcome is therefore the probability of entering the corresponding sector of indistinguishable states, which has nonzero operational thickness.

\begin{theorem}%[Normal distribution and the Born rule]
\label{thm:normal-to-born}
Under the identification of classical space $\mathbb R^3$ with the localized-state submanifold $M_3^\sigma$, the Born transition law between states in $M_3^\sigma$ reproduces the normal probability law for the measured position of the corresponding classical particle. Conversely, if transition probabilities in projective state space depend only on Fubini--Study distance, then the normal probability law on $M_3^\sigma$ determines a unique extension to arbitrary pairs of states. This extension is the Born transition law.
\end{theorem}

\begin{proof}[Sketch of proof]
The embedding
\[
\omega:\mathbb R^3\to \mathbb{CP}^{L_2},
\qquad
\omega({\bf a})=g_{{\bf a},\sigma},
\]
identifies the classical position \({\bf a}\) of a particle with the localized state \(g_{{\bf a},\sigma}\in M_3^\sigma\). Therefore, the probability of finding a Brownian particle in a small region \(W\subset\mathbb R^3\) may be equivalently viewed as the probability that the corresponding state lies in the region \(\omega(W)\subset M_3^\sigma\). Thus the normal probability distribution for measured positions in \(\mathbb R^3\) becomes a probability distribution for localized states on \(M_3^\sigma\).

Let \(W \subset\mathbb R^3\) be a small region of diameter \(\delta\) centered at \({\bf b}\). The state associated with this finite-resolution region may be represented by a narrow Gaussian \(g_{{\bf b},\delta}\). Direct computation gives
\[
|\langle g_{{\bf a},\sigma},g_{{\bf b},\delta}\rangle|^2
=
\left(\frac{2\sigma\delta}{\sigma^2+\delta^2}\right)^3
\exp\!\left[
-\frac{({\bf a}-{\bf b})^2}{2(\sigma^2+\delta^2)}
\right].
\]
Equivalently,
\[
|\langle g_{{\bf a},\sigma},g_{{\bf b},\delta}\rangle|^2
=
\cos^2\rho(g_{{\bf a},\sigma},g_{{\bf b},\delta}).
\]
As \(\delta\) becomes small, the right-hand side is the normal probability density centered at \({\bf a}\), multiplied by the corresponding small volume factor. Hence the normal probability of finding the classical particle near \({\bf b}\) agrees with the Born transition probability from \(g_{{\bf a},\sigma}\) to the localized state representing that region.

Finally, the Fubini--Study distances between Gaussian states \(g_{{\bf a},\sigma}\) and \(g_{{\bf b},\delta}\) range over the full interval of possible projective distances. Therefore, if transition probabilities in projective state space depend only on Fubini--Study distance, the rule determined on \(M_3^\sigma\) determines the rule for arbitrary normalized states \(\psi\) and \(\phi\). This gives
\[
P(\psi\to\phi)
=
\cos^2\rho(\psi,\phi)
=
|\langle\psi,\phi\rangle|^2,
\]
which is the Born rule.
\end{proof}

The same conclusion applies to the equivalence-class manifold
\(\widetilde M_3^\sigma\). The Gaussian representatives determine the induced
metric on \(\widetilde M_3^\sigma\), and distances between classes are computed
using these representatives. Hence a classical particle represented by an
equivalence class has the same distance relation to other classes as the
corresponding Gaussian representative has to the other Gaussian representatives.
The normal law for the particle position therefore agrees with the Born
transition law between the corresponding representatives. Since this gives the
same distance-dependent rule as on \(M_3^\sigma\), the extension to the full
projective state space proceeds exactly as before.

Dynamically, the result says that the same state-space diffusion has two complementary descriptions. When the isotropic walk is constrained to the localized classical sector, or equivalently to \(\widetilde M_3^\sigma\), it becomes the Brownian motion of a measured classical particle, producing the normal distribution of measurement errors. When the same isotropic walk is allowed to evolve through the full projective state space, the transition probabilities are governed by the Fubini--Study distance and therefore by the Born rule. Thus the Gaussian distribution of classical measurement errors and the Born rule of quantum measurement arise from the same geometric source: isotropic diffusion on state space. Classical measurements probe the process through its induced motion on the localized equivalence-class manifold, whereas microscopic measurements involve diffusion through the full projective state space before a detector-defined equivalence class is reached.

Having established the geometric relation between the normal law on the classical submanifold and the Born rule in projective state space, we now turn to the macroscopic regime. The goal is to show that the same \({\bf (RM)}\)-induced stochastic dynamics, combined with the tangent Schr\"odinger flow on \(M_{3,3}^{\sigma}\), yields stable Newtonian motion when conditioned on finite-resolution environmental records.

\section{Macroscopic classicality as a conditioned stochastic process}
\label{sec:classicality}

The previous sections described state reduction as an {\bf (RM)}-induced isotropic random walk on projective state space together with equivalence classes determined by finite detector resolution. We now apply the same mechanism to the motion of macroscopic bodies. The goal is to show, at the mathematical level, how Newtonian motion arises as a stroboscopic process: the state undergoes {\bf (RM)} diffusion in a neighborhood of the classical sector, while repeated returns to the localized sector produce recorded positions distributed narrowly around the Newtonian trajectory. In the next section, we provide physical estimates supporting the validity of these assumptions.

We begin with the one-dimensional case of a macroscopic point particle, since it contains the essential mechanism. As in Section \ref{equiv}, let \(M_{1,1}^{\sigma}\) denote the phase-space submanifold of localized particle states, with coordinates \((a,p)\), and let \(\widetilde M_{1,1}^{\sigma}\) denote the corresponding formulation in terms of equivalence classes at position resolution \(\sigma\). On \(M_{1,1}^{\sigma}\), the tangent component of the Schr\"odinger flow gives
\begin{equation}
\label{Newtonian-tangent-flow}
\frac{da}{dt}=\frac{p}{M},
\qquad
\frac{dp}{dt}=-\frac{dV}{da}.
\end{equation}
The same equations hold on the equivalence-class manifold
\(\widetilde M_{1,1}^{\sigma}\), since the variables \((a,p)\) label the
corresponding classes and the induced tangent dynamics is unchanged. Thus,
whenever the state lies on \(\widetilde M_{1,1}^{\sigma}\), or sufficiently
close to the set underlying \(\widetilde M_{1,1}^{\sigma}\), the free
Schr\"odinger velocity has a well-defined Newtonian tangent component.

Mathematically, the process we aim to derive is a random walk with intermittent conditioning. The {\bf (RM)} term produces isotropic random increments in projective state space, while the Schr\"odinger flow contributes Newtonian tangent drift whenever the state lies sufficiently close to the set of states underlying \(\widetilde M_{1,1}^{\sigma}\). When the path reaches the \(\sigma\)-localized sector \(\widetilde M_{1,1}^{\sigma}\), the particle behaves classically, its coordinate \(a\) is recorded as a classical position, and the subsequent walk is conditioned on this recorded value. The observed positions then form a sequence of conditional random variables centered on the Newtonian trajectory. The task is to show that physically reasonable choices of the time step and step variance allow the {\bf (RM)} and Schr\"odinger contributions to be separated while keeping these conditional distributions narrow on the resolution scale \(\sigma\). In doing so, we will work with the representative manifold \(M_{1,1}^\sigma\), keeping in mind that the Newtonian component of the motion can equivalently be obtained on the equivalence-class manifold \(\widetilde M_{1,1}^\sigma\).

\subsection{Alternating drift and {\bf (RM)} diffusion}

Let the total effective Hamiltonian during environmental monitoring be written schematically as
\begin{equation}
\label{combined-Hamiltonian}
\widehat h_{\rm tot}
=
\widehat h+\widehat h_{\rm RM},
\end{equation}
where \(\widehat h\) is the usual Hamiltonian of the particle and \(\widehat h_{\rm RM}\) is the random-matrix Hamiltonian representing the environmental interaction. The process is considered on a coarse-grained time scale on which \(\widehat h_{\rm RM}\) acts during short interaction windows, while \(\widehat h\) gives the deterministic Schr\"odinger evolution between them.

The desired separation is
\begin{equation}
\label{alternating-evolution}
e^{-\frac{i}{\hbar}(\widehat h+\widehat h_{\rm RM})dt}
\approx
e^{-\frac{i}{\hbar}\widehat h_{\rm RM}dt}
e^{-\frac{i}{\hbar}\widehat h dt},
\end{equation}
on the dynamically relevant \(\sigma\)-localized sector. This approximation is justified when the deterministic displacement and spreading produced by \(\widehat h\) during a single interaction window are negligible on the resolution scale. Equivalently, if \(P_\sigma\) denotes the local tangent projection onto the \(\sigma\)-localized tube around \(M_{1,1}^{\sigma}\), one requires
\begin{equation}
\label{commutator-small}
P_\sigma[\widehat h,\widehat h_{\rm RM}]P_\sigma=O(\varepsilon),
\qquad
\varepsilon\ll1.
\end{equation}
Under this condition the evolution may be treated, to the required accuracy, as alternating Newtonian tangent drift and {\bf (RM)} diffusion.

As shown in Section \ref{sec:classical-submanifolds}, for states on
\(M_{1,1}^{\sigma}\), the Schr\"odinger velocity admits the orthogonal
decomposition
\begin{equation}
\label{COM-decomposition}
-\frac{i}{\hbar}\widehat h\psi
=
\frac{da}{dt}\,\partial_a\psi
+
\frac{dp}{dt}\,\partial_p\psi
+
X_\perp,
\end{equation}
where the first two terms are tangent to \(M_{1,1}^{\sigma}\), while
\(X_\perp\) is orthogonal to \(M_{1,1}^{\sigma}\) with respect to the
Fubini--Study metric. The tangent coefficients are given by \eqref{Newtonian-tangent-flow}. For states lying on \(M_{1,1}^{\sigma}\), the free-evolution velocity thus decomposes into a tangent component generating classical drift and an orthogonal component responsible for spreading.

The orthogonal component is small for macroscopic bodies on environmental interaction time scales. For Gaussian representatives one obtains the spreading scale
\begin{equation}
\label{spreading-scale}
\|X_\perp\|\sim \frac{\hbar}{M\sigma^2},
\end{equation}
or equivalently the spreading time
\begin{equation}
\label{spreading-time}
T_{\rm spr}=\frac{M\sigma^2}{\hbar}.
\end{equation}
Thus the two basic smallness conditions are
\begin{equation}
\label{smallness-conditions}
\frac{v\,dt}{\sigma}\ll1,
\qquad
\frac{dt}{T_{\rm spr}}\ll1,
\end{equation}
where \(v\) is the characteristic classical velocity and \(dt\) is the duration of an environmental interaction window. The first condition says that the tangent drift during one interaction is below resolution; the second says that wave-packet spreading during one interaction is negligible.

\subsection{Conditioning on returns to the localized sector}

The {\bf (RM)} steps move the state away from the exact classical submanifold into the surrounding projective state space. The relevant question is therefore not whether the state remains exactly on the set underlying \(\widetilde M_{1}^{\sigma}\), but whether it returns frequently to this localized sector, where a position record can be assigned.

As in the previous section, let
\[
s=\ln(\delta_z/\sigma)
\]
be a transverse coordinate measuring the logarithmic width of the state, with \(\sigma\) a fixed reference length. The condition that the state lies in the localized sector is
\begin{equation}
\label{localized-sector}
s\le 0.
\end{equation}
Under {\bf (RM)} dynamics without drift, the discrete-time process \(s_n\), evaluated at successive {\bf (RM)} steps, is a symmetric random walk. Thus the return problem to \(\widetilde M_{1}^{\sigma}\) becomes a standard hitting problem for a one-dimensional random walk.

Let \(\tau_{\rm ret}\) be the first return time to the sector \(s\le0\). By standard recurrence properties of symmetric one-dimensional random walks, the return probability is nonzero and returns occur with high probability on sufficiently long time scales. More quantitatively, the Sparre Andersen theorem gives, for the probability that the walk has not returned to \(s\le0\) after \(n\) steps,
\begin{equation}
\label{Sparre-Andersen}
\mathbb P(\tau_{\rm ret}>n)
=
\frac{\binom{2n}{n}}{2^{2n}}
\sim
\frac{1}{\sqrt{\pi n}}.
\end{equation}
This estimate shows that frequent returns to the localized sector are typical once the {\bf (RM)} step time is sufficiently small.

Whenever the path reaches \(\widetilde M_1^\sigma\), the coordinate \(a\) is recorded. A particle whose state lies in the set underlying \(\widetilde M_{1}^\sigma\) behaves classically, and the recording process can be described by the usual Newtonian dynamics of the particle together with either the measuring device or its environment; see Section~\ref{sec:experiments}.
 The subsequent evolution is then conditioned on the recorded value. In this way, the process generates a sequence
\[
a_0,a_1,a_2,\ldots
\]
of recorded positions. These are not exact values of an underlying classical trajectory, but conditional random variables associated with repeated returns to the detector-defined localized sector and centered on the Newtonian trajectory.

\subsection{Stroboscopic Newtonian motion}
\label{strobo}

Let \(a_{\rm N}(t)\) denote the Newtonian solution generated by \eqref{Newtonian-tangent-flow}. The recorded positions \(a_k\) define stroboscopic Newtonian motion if their conditional distributions remain sharply concentrated around \(a_{\rm N}(t_k)\) on the resolution scale:
\begin{equation}
\label{stroboscopic-condition}
\mathbb E(a_k\mid a_{k-1})=a_{\rm N}(t_k)+o(\sigma),
\qquad
\operatorname{Var}(a_k\mid a_{k-1})\ll\sigma^2.
\end{equation}
The first condition expresses that the deterministic component of the motion is Newtonian. The second expresses that the {\bf (RM)}-induced stochastic deviations do not smear the recorded position beyond detector resolution.

The {\bf (RM)} contribution to the recorded coordinates may be written schematically as
\begin{equation}
\label{RM-increments}
\Delta a_{\rm RM}\sim \xi_a\sqrt{dt},
\qquad
\Delta p_{\rm RM}\sim \xi_p\sqrt{dt},
\end{equation}
where \(\xi_a\) and \(\xi_p\) are zero-mean random variables whose variances are determined by the diffusion coefficients \(D_a\) and \(D_p\) of the induced diffusion on 
\(M_{1,1}^\sigma\), respectively.
 Thus over a return interval \(\Delta t\), the stochastic spread in position is of order
\begin{equation}
\label{position-spread-rm}
\operatorname{Var}(\Delta a_{\rm RM})\sim D_a\,\Delta t.
\end{equation}
The recorded motion is Newtonian on the resolution scale provided
\begin{equation}
\label{diffusion-resolution-condition}
D_a\,\Delta t\ll\sigma^2,
\end{equation}
where \(\Delta t\) is the typical time between returns to \(\widetilde M_1^\sigma\).

\begin{theorem}[Stroboscopic Newtonian motion]
\label{thm:stroboscopic-newton}
Suppose that the {\bf (RM)} step time and step variances are chosen so that:
\begin{enumerate}
\item the free and {\bf (RM)} contributions separate on the \(\sigma\)-localized sector, i.e.
\[
P_\sigma[\widehat h,\widehat h_{\rm RM}]P_\sigma=O(\varepsilon),
\qquad
\varepsilon\ll1;
\]
\item the deterministic tangent displacement and spreading during one interaction window satisfy
\[
\frac{v\,dt}{\sigma}\ll1,
\qquad
\frac{dt}{T_{\rm spr}}\ll1;
\]
\item the {\bf (RM)}-induced stochastic position spread between returns satisfies
\[
D_a\,\Delta t\ll\sigma^2.
\]
\end{enumerate}
Then the sequence of recorded positions obtained by conditioning on returns to \(\widetilde M_1^\sigma\) is concentrated around the Newtonian trajectory on the resolution scale \(\sigma\). In this sense, the combined Schr\"odinger--{\bf (RM)} process yields stroboscopic Newtonian motion.
\end{theorem}

\begin{proof}[Sketch of proof]
Between {\bf (RM)} steps, the tangent component of the Schr\"odinger flow on \(\widetilde M_{1,1}^{\sigma}\) is Newtonian. The first two assumptions imply that, during each environmental interaction window, the deterministic tangent displacement is below resolution and the orthogonal spreading is negligible. The commutator estimate allows the total evolution to be treated as alternating free drift and {\bf (RM)} steps.

The {\bf (RM)} contribution has zero mean in the tangent coordinates and variance controlled by \(D_a\) and \(D_p\). Conditioning on returns to the localized sector produces recorded positions whose conditional mean follows the deterministic Newtonian drift, while their variance is bounded by the accumulated {\bf (RM)} diffusion between returns. The condition \(D_a\Delta t\ll\sigma^2\) ensures that this variance remains below detector resolution. Therefore the recorded positions form a stroboscopic sequence concentrated around the Newtonian trajectory.
\end{proof}

This theorem separates the mathematical and physical parts of the argument. Mathematically, if the time step, step variance, and return-time conditions hold, the conditioned process yields stroboscopic Newtonian motion. The remaining question is whether realistic environmental interactions for macroscopic bodies provide such parameters. This will be addressed in the next section.

\section{Physical estimates}
\label{sec:estimates}

The previous section separated the mathematical question from the physical one. Mathematically, the combined Schr\"odinger--{\bf (RM)} process yields stroboscopic Newtonian motion whenever the drift, spreading, diffusion, and return-time estimates satisfy the inequalities stated in Theorem~\ref{thm:stroboscopic-newton}. We now show that these inequalities can be satisfied for physically reasonable parameters. We first consider macroscopic bodies and then microscopic measurements.

\subsection{Macroscopic bodies}
\label{macro}

We now specialize the preceding discussion to a macroscopic body monitored by
its environment. The relevant interaction events may be taken to be scattering
events with environmental molecules, or with ambient radiation. These events
occur over very short time intervals and repeatedly transfer position
information to the surroundings. The resulting stochastic contribution is
modeled by the \({\bf (RM)}\) term, while the ordinary Schr\"odinger Hamiltonian
provides the deterministic tangent motion.

Consider the center of mass of a macroscopic body of mass \(M\), localized at
position resolution \(\sigma\). For definiteness, take
\begin{equation}
\label{macro-values}
M=10^{-6}\,\mathrm{kg},
\qquad
\sigma=10^{-6}\,\mathrm{m}.
\end{equation}
For air molecules at room temperature, a typical thermal speed is
\(v_{\mathrm{th}}\sim 5\times 10^2\,\mathrm{m/s}\). Taking a molecular
interaction length of order
\[
\ell\sim 10^{-9}\,\mathrm{m},
\]
the duration of a single collision is estimated as
\[
dt\sim \frac{\ell}{v_{\mathrm{th}}}
\sim
\frac{10^{-9}\,\mathrm{m}}{5\times 10^2\,\mathrm{m/s}}
\sim
2\times 10^{-12}\,\mathrm{s}.
\]
Thus environmental molecular collisions naturally lead to interaction windows
of order
\[
\tau\lesssim 10^{-12}\,\mathrm{s}.
\]

The corresponding free spreading time is
\begin{equation}
\label{macro-spreading-time}
T_{\mathrm{spr}}
=
\frac{M\sigma^2}{\hbar}
\sim
\frac{10^{-6}\cdot10^{-12}}{10^{-34}}
\sim
10^{16}\,\mathrm{s}.
\end{equation}
Therefore, during a single environmental interaction window,
\begin{equation}
\label{macro-spreading-ratio}
\frac{\tau}{T_{\mathrm{spr}}}\lesssim 10^{-28}
\end{equation}
is negligible. The orthogonal component of the free Schr\"odinger velocity,
which is responsible for spreading for states on \(M_{1,1}^{\sigma}\),
therefore does not move the state appreciably away from the localized sector
during a single collision.

The tangent drift is also below the resolution scale during such a short
interval. If \(v\) is a typical macroscopic velocity, say \(v\lesssim
1\,\mathrm{m/s}\), then
\begin{equation}
\label{macro-drift-ratio}
\frac{v\,\tau}{\sigma}
\lesssim
\frac{(1\,\mathrm{m/s})(10^{-12}\,\mathrm{s})}{10^{-6}\,\mathrm{m}}
=
10^{-6}.
\end{equation}
Thus the deterministic displacement during a single environmental collision is
\(\sigma\)-invisible. Combining \eqref{macro-spreading-ratio} and
\eqref{macro-drift-ratio}, the free evolution during one interaction window
differs from a scalar action on the dynamically relevant \(\sigma\)-localized
sector only by a small error
\begin{equation}
\label{epsilon-estimate}
\varepsilon
\sim
\frac{v\,dt}{\sigma}
+
\frac{dt}{T_{\mathrm{spr}}}
\lesssim
10^{-6}.
\end{equation}

Let \(P_\sigma\) denote the local tangent projection at the point under
consideration, where the relevant tangent space is that of the set underlying
\(\widetilde M_{1,1}^{\sigma}\) in the Hilbert space. For states in this
localized sector, the above estimate gives, locally and to leading order,
\begin{equation}
\label{free-scalar-action}
\widehat h\psi
=
E_\psi\psi
+
O(\varepsilon),
\end{equation}
where \(E_\psi\) is the scalar by which \(\widehat h\) acts on \(\psi\) to
leading order. Equivalently, over the short interaction interval considered
here, the free evolution acts on the localized sector as an overall phase up to
an error \(O(\varepsilon)\).

Let \(\widehat h_{\mathrm{RM}}\) denote the Hamiltonian in \({\bf (RM)}\). Since
the free evolution is scalar to leading order on the localized sector, its
commutator with the \({\bf (RM)}\) contribution is negligible there. In the
local projected sense just described, this may be written schematically as
\begin{equation}
\label{commutator-estimate-projected}
P_\sigma[\widehat h,\widehat h_{\mathrm{RM}}]P_\sigma
=
O(\varepsilon),
\qquad
\varepsilon\ll1.
\end{equation}
Thus, on the dynamically relevant localized sector, the free Schr\"odinger
evolution and the \({\bf (RM)}\) interaction may be treated, to the required
accuracy, as effectively separable over individual environmental collision
windows. This justifies describing the total evolution as alternating free
Schr\"odinger segments and \({\bf (RM)}\) kicks.

\subsection{Momentum diffusion and stochastic corrections}

The \({\bf (RM)}\) contribution produces stochastic increments in the classical coordinates when the state is close, in the Fubini--Study metric, to the set underlying \(\widetilde M_{1,1}^{\sigma}\). We write the increments over a short environmental kick time \(\tau\) schematically as
\begin{equation}
\label{RM-coordinate-increments}
\Delta a_{\mathrm{RM}}\sim \xi_a\sqrt{\tau},
\qquad
\Delta p_{\mathrm{RM}}\sim \xi_p\sqrt{\tau},
\end{equation}
where \(\xi_a\) and \(\xi_p\) are zero-mean random variables whose variances are determined by the diffusion coefficients \(D_a\) and \(D_p\).

For environmental molecular scattering, the momentum diffusion coefficient may be estimated from the collision rate and the typical momentum transfer per collision. If \(\Gamma\) denotes the collision rate and \(q\) the characteristic momentum transfer in one collision, then
\begin{equation}
\label{Dp-estimate}
D_p\sim \Gamma q^2.
\end{equation}
For air at room temperature and for a millimeter-scale body, representative values are
\[
\Gamma\sim 10^{22}\,\mathrm{s}^{-1},
\qquad
q\sim 2.5\times 10^{-23}\,\mathrm{kg\,m/s}.
\]
Thus
\begin{equation}
\label{Dp-numerical-estimate}
D_p
\sim
\Gamma q^2
\sim
10^{22}\left(2.5\times10^{-23}\right)^2
\sim
6\times10^{-24}\,\mathrm{kg^2m^2/s^3}.
\end{equation}

For a single environmental kick time \(\tau\sim10^{-12}\,\mathrm{s}\), the accumulated stochastic momentum increment is
\begin{equation}
\label{single-kick-momentum}
\Delta p_{\mathrm{RM}}(\tau)
\sim
\sqrt{D_p\tau}
\sim
\sqrt{(6\times10^{-24})(10^{-12})}
\sim
2.5\times10^{-18}\,\mathrm{kg\,m/s}.
\end{equation}
For \(M=10^{-6}\,\mathrm{kg}\) and \(v\sim1\,\mathrm{m/s}\), the classical momentum is
\[
p_{\mathrm{cl}}=Mv\sim10^{-6}\,\mathrm{kg\,m/s}.
\]
Hence, over a single kick,
\begin{equation}
\label{single-kick-momentum-ratio}
\frac{\Delta p_{\mathrm{RM}}(\tau)}{p_{\mathrm{cl}}}
\sim
\frac{2.5\times10^{-18}}{10^{-6}}
\sim
10^{-12}.
\end{equation}
This agrees with the estimate that an individual environmental kick is negligible compared with the macroscopic Newtonian momentum.

The corresponding position diffusion coefficient for the center of mass is of order
\begin{equation}
\label{Da-estimate}
D_a\sim \frac{D_p}{M^2}.
\end{equation}
For \(M=10^{-6}\,\mathrm{kg}\), this gives
\begin{equation}
\label{Da-numerical-estimate}
D_a
\sim
\frac{6\times10^{-24}}{10^{-12}}
\sim
6\times10^{-12}\,\mathrm{m^2/s}.
\end{equation}
Thus, for a macroscopic body, the large value of \(M\) strongly suppresses the induced position diffusion.

Over a renewal interval \(T_{\mathrm{ren}}\), the accumulated stochastic momentum correction has size
\begin{equation}
\label{stochastic-momentum-renewal}
\Delta p_{\mathrm{stoch}}(T_{\mathrm{ren}})
\sim
\sqrt{D_pT_{\mathrm{ren}}}.
\end{equation}
Using the renewal time scale estimated in the next subsection,
\[
T_{\mathrm{ren}}\sim 3\times10^{-4}\,\mathrm{s},
\]
one obtains
\begin{equation}
\label{stochastic-momentum-numerical}
\Delta p_{\mathrm{stoch}}(T_{\mathrm{ren}})
\sim
\sqrt{(6\times10^{-24})(3\times10^{-4})}
\sim
4\times10^{-14}\,\mathrm{kg\,m/s}.
\end{equation}
Therefore
\begin{equation}
\label{momentum-correction-small}
\frac{\Delta p_{\mathrm{stoch}}(T_{\mathrm{ren}})}{p_{\mathrm{cl}}}
=
\frac{\sqrt{D_pT_{\mathrm{ren}}}}{Mv}
\sim
\frac{4\times10^{-14}}{10^{-6}}
\sim
4\times10^{-8}
\ll1.
\end{equation}
Thus even over a full renewal interval, the accumulated stochastic momentum correction remains negligible compared with the macroscopic Newtonian momentum.

Similarly, the position variance accumulated between successive returns to the localized sector is
\begin{equation}
\label{position-variance-between-returns}
\operatorname{Var}(\Delta a_{\mathrm{RM}})
\sim
D_aT_{\mathrm{ren}}.
\end{equation}
With the same representative values,
\begin{equation}
\label{position-variance-numerical}
D_aT_{\mathrm{ren}}
\sim
(6\times10^{-12})(3\times10^{-4})
\sim
2\times10^{-15}\,\mathrm{m^2}.
\end{equation}
Thus the corresponding root-mean-square displacement is
\begin{equation}
\label{position-rms-numerical}
\sqrt{D_aT_{\mathrm{ren}}}
\sim
4\times10^{-8}\,\mathrm{m}.
\end{equation}
For the resolution scale \(\sigma=10^{-6}\,\mathrm{m}\), this satisfies
\begin{equation}
\label{position-diffusion-small}
D_aT_{\mathrm{ren}}\ll\sigma^2,
\end{equation}
or equivalently
\[
\sqrt{D_aT_{\mathrm{ren}}}\ll \sigma.
\]
Therefore the conditional distribution of the recorded position remains narrow on the detector resolution scale.

\subsection{Return to the localized sector}

The remaining question is whether the state returns to the localized sector
frequently enough. As in Section~\ref{sec:classicality}, let
\[
s=\ln(\delta_z/\sigma)
\]
be the logarithmic transverse coordinate measuring spread, with \(\sigma\)
fixed and the localized sector represented by
\[
s\le0.
\]
Under \({\bf (RM)}\) dynamics with negligible drift in \(s\), the discrete
process \(s_n\), evaluated at successive environmental kicks, is a symmetric
random walk.

Let \(\tau_{\mathrm{ret}}\) denote the first return time to \(s\le0\), measured
in the number of kicks. By the Sparre Andersen theorem,
\begin{equation}
\label{return-probability}
\mathbb P(\tau_{\mathrm{ret}}>n)
=
\frac{\binom{2n}{n}}{2^{2n}}
\sim
\frac{1}{\sqrt{\pi n}}.
\end{equation}
Thus the probability of not returning decreases as \(n^{-1/2}\). For a sufficiently
small environmental kick time \(\tau\), this gives frequent returns to the
localized sector on macroscopic observational time scales.

For the representative estimates used here, the environmental kick time was estimated above as
\[
\tau\sim 10^{-12}\,\mathrm{s}.
\]
A return probability of \(0.999968\), corresponding to a four-standard-deviation
level, requires
\[
1-\frac{1}{\sqrt{\pi n}}\approx 0.999968,
\]
and hence
\begin{equation}
\label{return-steps-estimate}
n
\approx
\frac{1}{\pi(1-0.999968)^2}
\approx
3.1\times10^8 .
\end{equation}
The corresponding renewal time is therefore
\begin{equation}
\label{return-time-estimate}
T_{\mathrm{ren}}
\sim
n\tau
\sim
(3.1\times10^8)(10^{-12}\,\mathrm{s})
\sim
3\times10^{-4}\,\mathrm{s}.
\end{equation}
Thus, with overwhelming probability, the state returns to the localized sector
on a sub-millisecond time scale. During such a time interval, the deterministic
Newtonian drift remains well defined, while the stochastic position variance
remains below \(\sigma^2\), as estimated above. Consequently, successive
recorded positions form a narrow stochastic band around the Newtonian trajectory.

\subsection{Macroscopic and microscopic regimes}

The estimates above show that the same Schr\"odinger--\({\bf (RM)}\) dynamics
has two distinct limiting regimes. For macroscopic bodies, the relevant
conditions are
\begin{equation}
\label{macro-classical-regime}
\frac{v\,\tau}{\sigma}\ll1,
\qquad
\frac{\tau}{T_{\mathrm{spr}}}\ll1,
\qquad
P_\sigma[\widehat h,\widehat h_{\mathrm{RM}}]P_\sigma=O(\varepsilon),
\quad
\varepsilon\ll1,
\end{equation}
where \(\tau\) is the environmental kick time. These conditions are combined
with
\[
\frac{\Delta p_{\mathrm{RM}}(\tau)}{p_{\mathrm{cl}}}\ll1,
\qquad
\frac{\Delta p_{\mathrm{stoch}}(T_{\mathrm{ren}})}{p_{\mathrm{cl}}}\ll1,
\qquad
D_aT_{\mathrm{ren}}\ll\sigma^2.
\]
Here \(T_{\mathrm{ren}}\) is the high-probability renewal time for return to the
localized sector, estimated above. These inequalities define a combined
drift--diffusion--resolution regime: the deterministic tangent displacement
during a single environmental kick is below resolution, the orthogonal spreading
is negligible, the free and \({\bf (RM)}\) contributions separate to the
required accuracy, and the stochastic corrections remain small between
successive returns to the localized sector.

\begin{theorem}%[Macroscopic classical regime]
\label{thm:macro-estimates}
For macroscopic masses, ordinary localization resolutions, and realistic
environmental interaction times, there exist physically reasonable \({\bf (RM)}\)
step-size and time-step parameters satisfying the conditions of
Theorem~\ref{thm:stroboscopic-newton}. Consequently, the combined
Schr\"odinger--\({\bf (RM)}\) process keeps the state of a macroscopic body
close to the localized sector associated with
\(\widetilde M_{1,1}^{\sigma}\) and produces stroboscopic Newtonian motion on
the resolution scale \(\sigma\).
\end{theorem}

\begin{proof}[Sketch of proof]
The drift and spreading estimates 
\eqref{macro-spreading-ratio} and \eqref{macro-drift-ratio} imply that, during individual environmental
interaction windows, the additional effect of the free evolution is negligible,
apart from the Newtonian tangent drift already accumulated between such windows. The commutator estimate
\eqref{commutator-estimate-projected} justifies treating the evolution as
alternating free segments and \({\bf (RM)}\) kicks. The diffusion and return
estimates \eqref{position-diffusion-small} and \eqref{return-probability} show
that stochastic corrections remain below detector resolution between returns to
the localized sector. The conclusion then follows from
Theorem~\ref{thm:stroboscopic-newton}.
\end{proof}

The microscopic measurement regime is different. If the measurement interval
\(\Delta t_{\mathrm{meas}}\) is short enough that the system Hamiltonian is
negligible compared with the \({\bf (RM)}\) interaction, then
\begin{equation}
\label{micro-measurement-rm}
e^{-\frac{i}{\hbar}(\widehat h+\widehat h_{\mathrm{RM}})\Delta t_{\mathrm{meas}}}
\approx
e^{-\frac{i}{\hbar}\widehat h_{\mathrm{RM}}\Delta t_{\mathrm{meas}}}.
\end{equation}
There is then no appreciable Newtonian displacement during the measurement
interval. The state undergoes \({\bf (RM)}\)-induced isotropic diffusion in
projective state space and is recorded in one of the detector-defined
equivalence classes.

\begin{theorem}%[Microscopic measurement regime]
\label{thm:micro-measurement}
If the measurement interaction time is short enough that the contribution of
\(\widehat h\) is negligible compared with \(\widehat h_{\mathrm{RM}}\), then
the measurement process is governed by isotropic diffusion in projective state
space. The probabilities of the detector-defined outcome classes are the Born
probabilities.
\end{theorem}

\begin{proof}[Sketch of proof]
Under the approximation \eqref{micro-measurement-rm}, the measurement evolution
is generated by \(\widehat h_{\mathrm{RM}}\). By
Theorem~\ref{thm:RM-isotropy}, this produces homogeneous isotropic diffusion in
projective state space. By Theorem~\ref{thm:normal-to-born}, the probabilities of
reaching detector-defined equivalence classes are the Born probabilities.
\end{proof}

Thus microscopic measurement and macroscopic classical motion are different
parameter regimes of the same dynamical model. In the microscopic regime, the
\({\bf (RM)}\) term dominates and produces Born-rule state reduction. In the
macroscopic regime, frequent environmental \({\bf (RM)}\) interactions produce
repeated returns to the localized sector, while the tangent component of the
Schr\"odinger flow supplies Newtonian drift. The distinction is therefore not a
difference in fundamental laws, but a difference in regime. We conclude that the
same linear Schr\"odinger dynamics, supplemented by \({\bf (RM)}\), yields
Born-rule state reduction for microscopic particles and Newtonian trajectories
for macroscopic ones.

\section{Standard quantum experiments and paradoxes}
\label{sec:experiments}

We now indicate how the preceding framework applies to several standard quantum experiments and foundational paradoxes. The purpose of this section is not to introduce new assumptions or prove additional theorems. Rather, we use the two ingredients established above---classical submanifolds of projective state space and {\bf (RM)}-induced stochastic unitary evolution with detector-defined equivalence classes---to reinterpret the usual paradoxes in a single language. In each case, the apparent difficulty comes from trying to assign classical properties to states that do not lie on the appropriate classical submanifold, or from identifying measurement outcomes with exact rays rather than with finite-resolution equivalence classes.

\subsection{Measurement and state reduction}

The measurement problem is the question of how linear Schr\"odinger evolution
can yield definite outcomes governed by the Born rule. In the present
framework, measurement is not represented by an additional projection postulate.
Instead, during interaction with a measuring device, the state undergoes
stochastic but unitary evolution generated by the \({\bf (RM)}\) Hamiltonian.
The resulting path in projective state space eventually enters one of the
detector-defined equivalence classes corresponding to an outcome.

The outcome is definite because, at a given measurement event, the state enters
one detector-defined equivalence class. The stochasticity of the \({\bf (RM)}\)
Hamiltonian determines which class is reached, while the homogeneity and
isotropy of the induced random walk give the Born probabilities, as explained
in Section~\ref{sec:rm}. More precisely, the Born probabilities are the
relative probabilities of reaching the possible outcome classes, conditioned on
the state reaching the relevant localized sector. Thus what is traditionally
called collapse is the approach of the state to an operationally defined
outcome sector in projective state space.

This process should be distinguished from the later recording of the outcome.
Recording occurs after the state has entered an appropriate equivalence class
and proceeds through ordinary macroscopic dynamics of the measuring device and
environment. No nonunitary collapse law is added; the underlying evolution
remains unitary for each realization of the random Hamiltonian.

The use of equivalence classes is essential. A detector does not determine an
exact ray in Hilbert space; it determines a finite-resolution outcome class.
This avoids the difficulty that a random path in a high-dimensional state space
would generally have zero probability of hitting a prescribed ray or exact
lower-dimensional submanifold. The relevant probability is the probability of
entering a detector-defined equivalence class, which represents all states
indistinguishable by the apparatus at the given resolution.

A superposition of alternatives does not represent several simultaneously
realized classical configurations. It represents a single point in projective
state space, generally lying away from the relevant classical submanifold. The
corresponding classical property becomes well defined only when the state
enters the appropriate detector-defined equivalence class. Thus the apparent
ambiguity of a superposition arises from assigning classical attributes to a
state that does not lie in the classical sector.

\subsection{Double-slit experiment}

The double-slit experiment asks how a particle can display interference when unobserved by the slits, yet appear localized when measured. In this framework, the question is reformulated in state-space terms. A localized particle state near emission may lie close to the classical submanifold \(\widetilde M_3^\sigma\). After interaction with the slit screen, with both slits open and no which-slit measurement, the state generally leaves \(\widetilde M_3^\sigma\) and evolves through the full projective state space. 
In fact, the state becomes a superposition of two localized states, representing the particle near each slit. The Fubini--Study distance from such a superposition to \(\widetilde M_3^\sigma\) is not small.
In that regime, there is no classical trajectory in \(\mathbb R^3\), and it is not meaningful to say that the particle passed through one slit or both. Instead, the path of the particle leaves the classical-space submanifold \(\widetilde M_3^\sigma\) and passes ``over" the screen in projective state space.

If no which-slit detector is present, the state arriving at the final screen contains coherent contributions from both slits. During propagation, the corresponding state functions spread and overlap, so their superposition contains interference terms. The interaction with the screen then induces \({\bf (RM)}\) evolution, and the state enters one of the localized equivalence classes corresponding to a spot on the screen. Repeating the experiment samples these classes with Born probabilities, producing the interference pattern.

If a which-slit detector is present, the {\bf (RM)} interaction occurs at the slits. The state then enters one of the slit-defined equivalence classes with Born probabilities, and the subsequent evolution is conditioned on the recorded outcome. The later distribution at the screen is no longer the Born distribution of the coherent two-slit state. After which-slit localization, each particle evolves from a single localized state and behaves classically to the relevant accuracy, so repeated runs produce an approximately normal distribution of spots rather than an interference pattern. Interference disappears because the state has already been reduced to a detector-defined class before reaching the screen.

Thus the particle does not need to be described either as passing through both slits as a classical object or as secretly choosing one slit before measurement. Its path is a path in projective state space. When the state lies near \(\widetilde M_3^\sigma\), particle-like localization is meaningful. When the state moves away from this localized sector, wave-like interference phenomena may appear. Thus the distance from \(\widetilde M_3^\sigma\) distinguishes the localized, particle-like regime from the delocalized, wave-like regime.

As explained in Section~\ref{many-p}, the slit screen can be viewed
geometrically as part of the composite system. In a full description, the
Hilbert space is the tensor product of the particle Hilbert space and the
Hilbert space describing the relevant degrees of freedom of the screen. When
both the particle and the screen are in localized classical states, the
corresponding joint state lies on a product classical configuration-space
sector. In the simplest two-body picture this sector is represented,
schematically, by
$
\widetilde M_3^{\sigma_p}\otimes \widetilde M_3^{\sigma_s},
$
where \(\sigma_p\) and \(\sigma_s\) are the resolution scales associated with
the particle and the screen degrees of freedom. At this level, the geometry is
the same as the geometry of two classical objects in a single copy of classical
space.

The slit screen then selects two spatially separated regions through which the
particle state may pass. When the interaction with the screen produces a
superposition of alternatives associated with the two slits, the joint state no
longer lies on the product classical sector
$
\widetilde M_3^{\sigma_p}\otimes \widetilde M_3^{\sigma_s}.
$
Equivalently, the particle component is no longer represented by a single
localized point of \(\widetilde M_3^{\sigma_p}\). The screen remains
macroscopic and localized, but the particle state has moved away from the
classical space submanifold. This is the state-space meaning of the statement
that, with both slits open and no which-slit measurement, the particle does not
pass through one definite slit as a classical object.

\subsection{Cloud-chamber and bubble-chamber tracks}

A cloud-chamber or bubble-chamber track appears to show a microscopic particle
following a classical trajectory. In the present framework, the track is not a
continuous microscopic path in \(\mathbb R^3\). It is a sequence of localized
records produced by repeated measurement-like interactions with the medium.

This situation is analogous to the macroscopic case discussed above. A
macroscopic body in a natural environment is continuously monitored by
surrounding molecules, radiation, and other degrees of freedom. The resulting
frequent environmental interactions keep its state close to the localized
classical sector and produce a stable Newtonian trajectory. In a cloud chamber
or bubble chamber, the measured particle has much smaller mass, but this is
compensated by the properties of the medium: the medium supplies frequent,
localized, and effectively amplifying interactions, such as ionization or
bubble formation. These interactions repeatedly return the state to
detector-defined position classes.

Between such interactions, the particle state evolves in projective state
space. Each sufficiently strong interaction with the medium defines
position-like equivalence classes corresponding to ionization or
bubble-formation events. The \({\bf (RM)}\) mechanism drives the state into one
of these classes, and the subsequent evolution is conditioned on the recorded
event. Repetition produces a sequence
\[
a_1,a_2,\ldots,a_N
\]
of localized records.

Let \(\ell_{\mathrm{rec}}\) denote the typical distance between successive
record-forming events in the medium, and let \(u\) be the speed of the charged
particle. For a visible track, the relevant scale is microscopic; one may take
representative values
\[
\ell_{\mathrm{rec}}\sim 10^{-6}\text{--}10^{-5}\,\mathrm{m},
\qquad
u\sim 10^7\text{--}10^8\,\mathrm{m/s}.
\]
The time between successive localized records is then
\begin{equation}
\label{chamber-record-time}
T_{\mathrm{rec}}
\sim
\frac{\ell_{\mathrm{rec}}}{u}
\sim
10^{-14}\text{--}10^{-12}\,\mathrm{s}.
\end{equation}
Thus, although the particle is microscopic, the medium provides repeated
localization on extremely short time scales.

The relevant comparison is with the free spreading time
\[
T_{\mathrm{spr}}=\frac{m\sigma^2}{\hbar},
\]
where \(m\) is the particle mass and \(\sigma\) is the localization resolution
of the record. For an electron and a proton one finds, for
\(\sigma=10^{-6}\,\mathrm{m}\),
\[
T_{\mathrm{spr}}^{(e)}
\sim
\frac{(9\times10^{-31})(10^{-12})}{10^{-34}}
\sim
10^{-8}\,\mathrm{s},
\]
and
\[
T_{\mathrm{spr}}^{(p)}
\sim
\frac{(1.7\times10^{-27})(10^{-12})}{10^{-34}}
\sim
10^{-5}\,\mathrm{s}.
\]
Consequently,
\begin{equation}
\label{chamber-spreading-ratio}
\frac{T_{\mathrm{rec}}}{T_{\mathrm{spr}}}
\ll 1
\end{equation}
for both electrons and heavier charged particles. The spreading between
successive record-forming interactions is therefore negligible on the
resolution scale of the track.

Between record-forming interactions, the \({\bf (RM)}\) Hamiltonian is absent,
and the state evolves under the ordinary Schr\"odinger Hamiltonian. The only
separation estimate needed is therefore during the short ionization or
bubble-formation event itself, when the \({\bf (RM)}\) interaction is active.
Let \(\tau_{\mathrm{int}}\) denote the duration of such an event. If
\(\ell_{\mathrm{int}}\) is the microscopic interaction length and \(u\) is the
particle speed, then
\[
\tau_{\mathrm{int}}
\sim
\frac{\ell_{\mathrm{int}}}{u}.
\]
Taking
\[
\ell_{\mathrm{int}}\sim 10^{-10}\text{--}10^{-9}\,\mathrm{m},
\qquad
u\sim 10^7\text{--}10^8\,\mathrm{m/s},
\]
gives
\[
\tau_{\mathrm{int}}
\sim
10^{-18}\text{--}10^{-16}\,\mathrm{s}.
\]
During such a short interaction window,
\[
\varepsilon_{\mathrm{ch}}
\sim
\frac{u\tau_{\mathrm{int}}}{\sigma}
+
\frac{\tau_{\mathrm{int}}}{T_{\mathrm{spr}}}
\sim
\frac{\ell_{\mathrm{int}}}{\sigma}
+
\frac{\tau_{\mathrm{int}}}{T_{\mathrm{spr}}}
\ll1.
\]
For a localization resolution on the micron scale, \(\sigma\sim10^{-6}\,\mathrm m\),
the first term is \(10^{-4}\text{--}10^{-3}\), while the second is negligible
relative to the spreading times estimated above. Thus, during the interaction
window, the free Schr\"odinger evolution acts on the localized sector only up
to a projectively irrelevant phase and a negligible error. In the local
projected sense used above,
\[
P_\sigma[\widehat h,\widehat h_{\mathrm{RM}}]P_\sigma
=
O(\varepsilon_{\mathrm{ch}}).
\]
Thus the chamber dynamics may be described as alternating ordinary
Schr\"odinger evolution between records and short \({\bf (RM)}\)-dominated
kicks during record formation.

The deterministic tangent displacement between records is
\[
\Delta a_{\mathrm{cl}}\sim uT_{\mathrm{rec}}\sim \ell_{\mathrm{rec}}.
\]
Thus the observed track is naturally stroboscopic: each record is displaced
from the preceding one by the classical tangent motion during the short time
between interactions. If the recording resolution is comparable to the spacing between successive
record-forming events,
\[
\sigma \sim \ell_{\mathrm{rec}},
\]
with
\[
\ell_{\mathrm{rec}}\sim 10^{-6}\text{--}10^{-5}\,\mathrm{m},
\]
then each segment of the track is resolved as a localized classical record,
while the sequence of records follows the Newtonian trajectory to within the
spatial resolution of the chamber.

The momentum disturbance caused by a single record-forming interaction is also
small compared with the momentum of a typical charged particle in a chamber.
For a particle with momentum in the range
\[
p_{\mathrm{cl}}\sim 10^{-20}\text{--}10^{-19}\,\mathrm{kg\,m/s},
\]
corresponding to relativistic or semi-relativistic charged-particle tracks,
and for an ionization-scale energy transfer
\[
\Delta E\sim 10^1\text{--}10^3\,\mathrm{eV},
\]
the associated momentum transfer is of order
\[
q\sim \frac{\Delta E}{u}
\lesssim 10^{-25}\text{--}10^{-23}\,\mathrm{kg\,m/s}.
\]
Hence
\begin{equation}
\label{chamber-momentum-ratio}
\frac{q}{p_{\mathrm{cl}}}
\ll1.
\end{equation}
Thus the record-forming interactions localize the particle strongly enough to
produce a track, while perturbing the tangent Newtonian motion only weakly from
one record to the next.

%When the drift and diffusion parameters fall in this regime, the records are strongly correlated and lie close to a classical path. The observed track is therefore a stroboscopic record of repeated localization events, not evidence that the microscopic state possessed a sharply defined classical position at every intermediate time. The apparent classicality of the track arises from the same mechanism as macroscopic classicality, but with the role of large mass partly replaced by the high resolving and amplifying power of the surrounding medium.

The estimates in this subsection are only order-of-magnitude estimates. A
detailed chamber-specific calculation would require the ionization or
bubble-formation rate, the localization scale of each record, the momentum
transfer distribution, and the thermodynamic properties of the medium. The
purpose here is to show that the required regime is physically plausible: even
for microscopic particles, sufficiently frequent localized interactions with an
amplifying medium can return the state to the localized sector often enough
that the recorded sequence is effectively Newtonian.

\subsection{Stern--Gerlach measurement}

In a Stern--Gerlach experiment, the magnetic field correlates spin alternatives
with spatially separated wave packets. If the incoming spin state is
\[
\alpha\ket{+}+\beta\ket{-},
\]
the magnetic field produces a state of the form
\[
\alpha\ket{+}\ket{\phi_+}
+
\beta\ket{-}\ket{\phi_-},
\]
where \(\phi_+\) and \(\phi_-\) are directed toward distinct detector regions.

The magnetic field separates the alternatives, but the outcome is produced only
when the state interacts with the detector. At that stage, the two spatial
packets correspond to disjoint detector-defined equivalence classes. The
\({\bf (RM)}\) Hamiltonian acts on the full entangled position--spin state, and
therefore its random step may include components in both the position and spin
degrees of freedom. However, the equivalence classes relevant to the recorded
outcome are not spin equivalence classes. The apparatus records the localized
position of the particle, or the detector region in which it is found, and is
sensitive to spin only through its prior correlation with position.

Thus the same \((\tau,s)\)-description used above applies. The coordinate
\(\tau\) labels the recorded position, while \(s\) measures localization in the
position variable. Motion in spin directions, or in other directions that do
not change the recorded position or the localization width, lies along the
leaves of the corresponding foliation and does not by itself determine the
recorded outcome. The \({\bf (RM)}\)-induced reduction is therefore still a reduction to one of the position-defined detector classes. The Born rule then follows as before, applied to these detector-defined position outcomes.

Consequently, the Stern--Gerlach apparatus converts spin alternatives into
position-like detector outcomes. Since the spatial packets \(\phi_+\) and
\(\phi_-\) are correlated with the spin states \(\ket{+}\) and \(\ket{-}\), the
probabilities of the two detector records are
\[
P_+=|\alpha|^2,
\qquad
P_-=|\beta|^2.
\]
The preferred outcome basis is therefore determined by the measurement
arrangement, which correlates the spin degree of freedom with macroscopically
distinguishable position records.

\subsection{Macroscopic superpositions and Schr\"odinger's cat}

Schr\"odinger's cat paradox raises the question of how quantum theory can allow a macroscopic system to be in a superposition of classically distinct states, such as ``alive" and ``dead."  In the present framework, such a superposition corresponds to a state with components near different macroscopic classical sectors. Under ordinary conditions, a macroscopic system is continually monitored by its environment. Therefore the macroscopic analysis of Section~\ref{sec:estimates} applies to the cat. The \({\bf (RM)}\) mechanism rapidly drives the state into one of the corresponding detector- or environment-defined equivalence classes, while repeated environmental interactions keep returning the state to localized macroscopic sectors.

The cat is not observed in a superposition because the macroscopic degrees of freedom are continually recorded by the environment. These environmental interactions define robust equivalence classes corresponding to stable macroscopic records. Once one such class is reached, subsequent evolution is conditioned on it, and the record persists.
Thus the paradox arises from applying the linear superposition principle without accounting for the \({\bf (RM)}\)-environmental dynamics and the finite-resolution equivalence classes that define macroscopic records. The framework does not require a fundamental nonunitary collapse of the universal state, nor does it require many simultaneously realized cats. It gives a single recorded macroscopic outcome through stochastic unitary dynamics in state space.

\subsection{Measured system and measuring device as a composite system}

In a complete description of a measurement, the measured particle and the
macroscopic measuring device may be treated as a single composite system. This
composite description must agree with the simpler particle-centered description,
in which the measuring device is treated as the source of the \({\bf (RM)}\)
interaction and the induced diffusion is described as acting effectively on the
particle state.

We now show that the \({\bf (RM)}\) walk on the state space of the composite
system provides such a consistent description. The Brownian-lift argument in
Theorem~\ref{thm:Brownian-lift-RM} applies equally to composite systems. Thus
the Gaussian random walk representing Brownian motion of the classical
particle--device pair lifts to a unitary random walk in the projective state
space of the tensor-product Hilbert space. In a finite-dimensional
coarse-grained sector, this lift is again described by \({\bf (RM)}\).
Therefore the random Hamiltonian in \({\bf (RM)}\) may be applied to the full
particle--device state, not only to the state of the particle considered
separately.

Let \(a\) denote the particle position coordinate and \(A\) the relevant
macroscopic pointer, or center-of-mass, coordinate of the device. On the product
of the localized classical sectors \(M_1^{\sigma_p}\otimes M_1^{\sigma_d}\),
Gaussian representatives have the form
\[
g_{a,\sigma_p}\otimes G_{A,\sigma_d},
\]
where \(\sigma_p\) and \(\sigma_d\) are the resolution scales associated with
the particle and device degrees of freedom. For two such representatives,
\[
\left|
\left\langle
g_{a,\sigma_p}\otimes G_{A,\sigma_d},
g_{b,\sigma_p}\otimes G_{B,\sigma_d}
\right\rangle
\right|^2
=
\exp\!\left[
-\frac{(a-b)^2}{4\sigma_p^2}
-\frac{(A-B)^2}{4\sigma_d^2}
\right].
\]
Hence, to leading order, the induced Fubini--Study metric on the product
classical sector is
\begin{equation}
\label{product-metric-particle-device}
ds^2
=
\frac{da^2}{4\sigma_p^2}
+
\frac{dA^2}{4\sigma_d^2}.
\end{equation}

Introduce the Fubini--Study-normalized coordinates
\[
u=\frac{a}{2\sigma_p},
\qquad
v=\frac{A}{2\sigma_d}.
\]
Then \eqref{product-metric-particle-device} becomes
\[
ds^2=du^2+dv^2.
\]
Homogeneous and isotropic Brownian motion on the product classical sector means
that, over a small time interval \(dt\), the increments satisfy
\[
du\sim N(0,D\,dt),
\qquad
dv\sim N(0,D\,dt),
\qquad
\mathbb E[du\,dv]=0.
\]
Equivalently, in the original Euclidean coordinates,
\[
\operatorname{Var}(da)=4\sigma_p^2D\,dt,
\qquad
\operatorname{Var}(dA)=4\sigma_d^2D\,dt,
\qquad
\mathbb E[da\,dA]=0.
\]
Thus isotropy is imposed in the orthonormal Fubini--Study coordinates
\((u,v)\), not in the Euclidean coordinates \((a,A)\).

The device scale \(\sigma_d\) is fixed operationally by the Brownian uncertainty
of the macroscopic pointer coordinate over the relevant observation time. Let
\(D_P^{(d)}\) denote the momentum-diffusion coefficient for the device. The
corresponding effective displacement scale of the pointer or center-of-mass
coordinate is suppressed by the large device mass; schematically,
\[
D_A^{(d)}
\sim
\frac{D_P^{(d)}}{M_d^2},
\]
where \(D_A^{(d)}\) denotes the induced position-diffusion scale over the
measurement interval. Thus, over the measurement time, the device remains
localized within its own macroscopic equivalence class.

Under these assumptions, the Brownian-lift argument applies exactly as in
Theorem~\ref{thm:Brownian-lift-RM}. The classical product sector
\(M_1^{\sigma_p}\otimes M_1^{\sigma_d}\) is a real submanifold of the
projective state space of the tensor-product Hilbert space. If the unitary lift
of the Gaussian random walk on this sector is required to be homogeneous and
isotropic in the Fubini--Study metric of the composite projective state space,
then, in each finite-dimensional coarse-grained sector, the corresponding
random Hermitian generator has the centered unitarily invariant Gaussian law.

In this formulation the lift is unique up to the overall diffusion scale. The
resolution parameters \(\sigma_p\) and \(\sigma_d\) fix the metric
identification between Euclidean coordinates and Fubini--Study coordinates,
while the scalar diffusion coefficient \(D\) fixes the variance of the
isotropic Brownian motion in the normalized coordinates \((u,v)\). Once these
data are fixed, the homogeneous and isotropic unitary lift is the \({\bf (RM)}\)
walk on the finite-dimensional coarse-grained composite state space.

We now use this to relate the composite \({\bf (RM)}\) walk to the effective
particle-centered description. For the macroscopic device, the induced Brownian
displacement of the pointer coordinate \(A\) during the measurement interval is
below the resolution \(\sigma_d\). Thus the Brownian motion of the device coordinate remains, with overwhelming
probability, within the same resolution cell corresponding to the recorded
macroscopic value of \(A\). In state-space language, this means that the
device state remains in the same macroscopic equivalence class, denoted
\(\{\Psi\}\), up to corrections below the device resolution.

In more detail, let
\[
\Phi_0=\varphi\otimes\Psi
\]
be an initial particle--device state, where \(\Psi\) is a representative of the
macroscopic device equivalence class \(\{\Psi\}\) associated with the recorded
value of the pointer coordinate \(A\).
 Let \(Q\) denote the coarse-grained
projection onto the device sector represented by the equivalence class
\(\{\Psi\}\), and let \(Q^\perp=I-Q\). Since the induced Brownian motion of the
macroscopic device coordinate remains in \(\{\Psi\}\) with overwhelming
probability during the measurement interval, one has, after one {\bf (RM)}
step,
\[
\|(I\otimes Q^\perp)\Phi_1\|^2\ll1.
\]
Thus the component of \(\Phi_1\) in which the device lies outside its
macroscopic equivalence class has negligible weight. Equivalently,
\[
\Phi_1=(I\otimes Q)\Phi_1+O(\varepsilon),
\qquad
\varepsilon\ll1.
\]
At the level of detector-defined equivalence classes this gives
\[
\{\Phi_1\}\simeq \{\varphi_1\}\otimes \{\Psi\},
\]
up to corrections below the device resolution.

Repeating the same argument for the finite sequence of steps involved in the
measurement process, and choosing the step size small enough that the
accumulated weight outside the device class remains negligible, gives
\[
\{\Phi_k\}\simeq \{\varphi_k\}\otimes \{\Psi\}
\]
at each step \(k\) of the process, until the device records a new macroscopic
outcome class. Thus the full composite \({\bf (RM)}\) dynamics preserves the
effective product form in the operational sense relevant to measurement: the
device remains in its macroscopic record class, while the particle undergoes
the effective stochastic evolution leading to one of the detector-defined
position equivalence classes.

Consequently, the complete composite-system description reduces, at the
operational level, to the particle-centered description. The \({\bf (RM)}\)
Hamiltonian acts on the full particle--device state, but the macroscopic device
remains in its classical record class, while the microscopic particle undergoes
effective \({\bf (RM)}\)-induced reduction to one of the detector-defined
position classes. When the particle state lies on, or sufficiently near, its
localized phase-space manifold, its interaction with the device is described by
the Newtonian tangent dynamics derived above. The device then records the
corresponding classical position through ordinary macroscopic dynamics.

\subsection{Wigner's friend}

The Wigner's friend scenario asks how the friend can record a definite outcome
while Wigner may describe the larger laboratory quantum mechanically. In the
present framework, the relevant system is composite: it includes the microscopic
particle, the measuring apparatus, the friend, and the surrounding laboratory.
The analysis of the particle--device system above applies here as well, with
the friend and laboratory included among the macroscopic degrees of freedom.

Before the friend's measurement, the relevant state may be written schematically as
\[
\Phi_0=\varphi\otimes \Psi_A\otimes \Psi_F\otimes \Psi_L,
\]
where \(\varphi\) is the state of the microscopic particle, while
\(\Psi_A,\Psi_F,\Psi_L\) represent the macroscopic states of the apparatus, the
friend, and the laboratory environment. During the measurement, the
\({\bf (RM)}\) dynamics acts on this full composite state. However, as in the
particle--device case, the macroscopic factors remain in their classical
equivalence-class sectors, while the microscopic alternative being measured is
correlated with one of the detector-defined outcome classes.

If the possible outcomes are labeled by \(c\), then after the friend's
measurement the realized state belongs, at the level of equivalence classes, to
one class of the form
\[
\{g_{c,\sigma_p}\}\otimes
\{\Psi_A^{(c)}\}\otimes
\{\Psi_F^{(c)}\}\otimes
\{\Psi_L^{(c)}\}.
\]
Here \(\{g_{c,\sigma_p}\}\) is the particle outcome class and
\(\{\Psi_A^{(c)}\}\), \(\{\Psi_F^{(c)}\}\), and \(\{\Psi_L^{(c)}\}\) are the
corresponding macroscopic record classes of the apparatus, friend, and
laboratory. The friend's accessible information is precisely this macroscopic
record.

This is analogous to an ordinary classical measurement. If a person measures a
classical particle in a laboratory, the outcome is recorded by the measuring
device and by the person. Adding a friend, or later an outside observer, does
not change the already recorded value; it only correlates additional
macroscopic systems with the same record. The present framework gives the same
structure in state space. The difference is that the microscopic particle is
described quantum mechanically until its state enters one of the
detector-defined equivalence classes, whereas the apparatus, friend, and
laboratory remain in the macroscopic classical regime throughout.

When Wigner subsequently interacts with the laboratory, he becomes part of the
same composite measurement process. If \(\Psi_W\) denotes Wigner's initial
macroscopic state, then the relevant state before his observation is
schematically
\[
\Phi_W=
g_{c,\sigma_p}\otimes
\Psi_A^{(c)}\otimes
\Psi_F^{(c)}\otimes
\Psi_L^{(c)}\otimes
\Psi_W,
\]
at the level of the realized equivalence class. The interaction with the
laboratory correlates Wigner with the already established record, giving
\[
\{g_{c,\sigma_p}\}\otimes
\{\Psi_A^{(c)}\}\otimes
\{\Psi_F^{(c)}\}\otimes
\{\Psi_L^{(c)}\}\otimes
\{\Psi_W^{(c)}\}.
\]
Thus Wigner records the same value \(c\). He does not create a new outcome, nor
does his observation erase or alter the macroscopic record already formed
inside the laboratory.

In this sense, the final situation is again a particle--device system, but with
a larger macroscopic device: the apparatus, friend, Wigner, and laboratory environment together form one large macroscopic record system. The measured
particle is the microscopic component, while the rest of the composite system
remains in the classical equivalence-class regime. Thus the friend and Wigner
do not obtain incompatible facts; they become correlated with the same
state-space process at different stages.

The distinction is between the full unitary description of the enlarged system
and the equivalence-class description of accessible macroscopic records.
Definite records are not observer-relative in the sense of being created by
subjective knowledge. They are objective features of the state entering a
detector- or environment-defined sector. The consistency of the accounts is due
to the fact that macroscopic objects, under {\bf (RM)}, remain in the
classical equivalence-class regime and change only through effectively
classical correlations. Once the measured particle has entered a localized
outcome class, the apparatus, friend, Wigner, and laboratory environment are
correlated with that same class as macroscopic record systems. Thus the
agreement between the friend and Wigner is, at that stage, an ordinary
classical correlation between records.

%%%%

\subsection{EPR correlations and nonlocality}

EPR and Bell-type experiments demonstrate that spatially separated systems can
display correlations stronger than any local hidden-variable model allows. In
the present framework, an entangled pair is represented by a single point in
the joint projective state space, not by two independent points in ordinary
space. As in the particle--device case considered above, the \({\bf (RM)}\)
random walk acts on the joint state, producing a state-space trajectory. The
difference is that here both components of the pair are microscopic, so neither
component is fixed in a macroscopic record class before measurement.

For a two-particle state in one spatial dimension, the Hilbert space is
$
L_2(\mathbb R)\otimes L_2(\mathbb R),
$
and the relevant projective state space is the projectivization of this
tensor-product Hilbert space. The classical configuration-space sector
$
\widetilde M_1^{\sigma_1}\otimes \widetilde M_1^{\sigma_2}
$
is represented by products of localized equivalence classes of the form
\[
\{g_{c,\sigma_1}\}\otimes \{g_{d,\sigma_2}\},
\]
corresponding to the classical situation of two particles with positions \(c\)
and \(d\). Here \(\sigma_1\) and \(\sigma_2\) are the corresponding detector
resolutions.

An entangled initial state is a superposition
of alternatives associated with these classes,
\[
\Phi
=
\sum_k \alpha_k\,
g_{c_k,\sigma_1}\otimes g_{d_k,\sigma_2}.
\]
Such a state does not lie on the classical configuration-space submanifold
$
\widetilde M_1^{\sigma_1}\otimes \widetilde M_1^{\sigma_2},
$
except in the special case in which the two particles already have definite
localized positions.

Measurement corresponds to the joint state entering one equivalence class
associated with a pair of detector outcomes,
$
\{g_{c,\sigma_1}\}\otimes \{g_{d,\sigma_2}\}.
$
The Born rule for the joint state is derived from the {\bf (RM)} walk in
the same way as for a single measured system, but now applied to the joint
projective state space. It gives the usual correlations between the recorded
outcomes. Thus, if the measurement process brings the first particle into the
class \(\{g_{c,\sigma_1}\}\), then the joint state enters the corresponding
product class, and the second particle is recorded in the correlated class
\(\{g_{d,\sigma_2}\}\). The correlated outcome is therefore not produced by a
signal or physical influence propagating through classical space. It follows
from the motion of the joint state toward a single detector-defined equivalence
class of the pair.

Bell's theorem constrains local hidden-variable models formulated in spacetime.
The present framework is not such a model. The correlations arise from the
position of the initial state in state space, from the geometry of the classical
configuration-space submanifold
$
\widetilde M_1^{\sigma_1}\otimes \widetilde M_1^{\sigma_2},
$
and from the Born-rule transition probabilities generated by \({\bf (RM)}\).
The evolution toward
$
\widetilde M_1^{\sigma_1}\otimes \widetilde M_1^{\sigma_2}
$
is unitary and local in time. It is also local in state space: infinitesimal
time increments produce infinitesimal displacements in the Fubini--Study
metric. Relativistic no-signaling is preserved because the marginal
probabilities remain those of standard quantum mechanics, so the correlations
cannot be used to transmit information superluminally.

In the single-particle case, wave-like properties were associated with the
particle's state being away from a classical space manifold, such as
\(\widetilde M_1^{\sigma}\). Similarly, for an entangled pair, what appears as
nonlocality in classical space reflects the fact that the joint state is a
single point in the projective state space of
\(L_2(\mathbb R)\otimes L_2(\mathbb R)\), lying outside the classical
configuration-space submanifold
$
\widetilde M_1^{\sigma_1}\otimes \widetilde M_1^{\sigma_2}.
$
Accordingly, EPR correlations and Bell-inequality violations are reinterpreted
as geometric features of state space, its classical submanifolds, and the
stochastic unitary dynamics generated by \({\bf (RM)}\), rather than as evidence
for nonlocal dynamics in classical space.

\subsection{Preferred basis and classical observables}

The preferred-basis problem is the question of why measurements yield outcomes
in particular bases, especially position, rather than in arbitrary superpositions. In the
present framework, the \({\bf (RM)}\) dynamics itself is homogeneous and
isotropic in projective state space and therefore does not select a preferred
basis by hand. The relevant basis is selected by the measurement arrangement,
by the observable-related classical submanifold, and by the equivalence relation
induced by finite detector resolution.

Position has a distinguished role because macroscopic measuring devices
ultimately record positions or configurations of matter. Other observables are
measured by correlating them with position through the design of the apparatus.
For example, a Stern--Gerlach device correlates spin alternatives with spatially
separated detector regions, and a spectrometer correlates momentum values with
positions on a detection screen. Thus the measured observable is determined by
the apparatus through a mapping from the alternatives of that observable to
detector-defined position sectors.

Equivalently, each measurement defines a family of outcome sectors in state
space. For position measurement, these sectors are the localized equivalence
classes associated with a classical space submanifold, such as
\(\widetilde M_1^\sigma\). For other observables, the corresponding submanifold
is obtained from the physical and mathematical transformation relating that
observable to position records. For instance, the manifold of approximate
momentum eigenstates may be obtained from the position manifold by applying the
Fourier transform.

In this sense, a preferred basis is not a separate postulate. It is determined
by the geometry of the relevant outcome sectors, the design of the measuring
device, and the equivalence classes of states indistinguishable at the detector
resolution. An observable becomes classical precisely when the state reaches
one of the corresponding detector-defined equivalence classes, so that a
macroscopic record can be formed.

\subsection{Quantum Zeno effect}

The quantum Zeno effect is usually described as the inhibition of evolution by
frequent measurements, often formulated in terms of repeated projection. In the
present framework, no separate projection postulate is required to account for
the effect. Repeated measurements correspond to repeated \({\bf (RM)}\)-induced
returns of the state to the same detector-defined equivalence class, or to a
small neighborhood of that class. Each return conditions the subsequent
evolution on the recorded class.

As the interval between measurement interactions decreases, the probability
that the state escapes the corresponding operational neighborhood becomes
small. The effect traditionally attributed to repeated projection is therefore
described as the dynamical stabilization of a state-space path by repeated
stochastic unitary interactions and conditioning on records. In this sense, the
process that returns the state to the equivalence class plays the role usually
assigned to projection.

As explained in the next section, recording does not introduce a new dynamical
law. It only fixes the conditional information used to describe subsequent
evolution. The suppression of transitions follows from the \({\bf (RM)}\)
dynamics, the geometry of the relevant equivalence class, and the repeated
conditioning on recorded returns to that class.

\subsection{Recording}
\label{recording}

It is important to distinguish state reduction from the recording of an
outcome. In the present framework, collapse is the dynamical approach of the
state to a detector-defined equivalence class through the \({\bf (RM)}\)
process. For a position measurement, this means that the state reaches a class
such as \(\{g_{c,\sigma}\}\), or the corresponding phase-space class when
momentum information is also relevant. At that point, the particle has a
well-defined position at the resolution of the measuring device.

The subsequent recording of this fact is not an additional collapse process.
Once the state belongs to a localized position equivalence class, the situation
has entered the classical regime relevant to the detector. Namely, as explained
in Section~\ref{many-p}, the Schr\"odinger dynamics of the system becomes
equivalent to its Newtonian dynamics. The particle, measuring device, and
nearby environment can then interact through ordinary macroscopic dynamics. A
pointer may move, an atom may be ionized, a bubble may form, a scintillation
event may occur, or a stable mark may be produced. These are classical
correlations between an already localized particle state and macroscopic
degrees of freedom of the apparatus.

Thus the detector does not create localization merely by recording the result.
Rather, the \({\bf (RM)}\)-induced state-space dynamics brings the state to a
detector-defined equivalence class, and the detector records that this has
happened. The recording process is then governed by the Newtonian dynamics of
the particle--device--environment system, because the relevant state is already
in, or sufficiently close to, the classical sector.

This also explains why recording is stable. The apparatus is macroscopic and is
continually monitored by its environment. Its state therefore remains in a
macroscopic equivalence-class sector, in close analogy with the dynamical
stabilization described in connection with the quantum Zeno effect, and the
recorded outcome becomes encoded in many correlated degrees of freedom. Once
such a record is formed, subsequent evolution is conditioned on it. The record
supplies the initial data for the next stage of the effective stochastic
evolution, but it does not introduce a new dynamical law.

In this sense, recording is a classical amplification and stabilization of an
outcome that has already been selected dynamically in state space. Collapse is
the approach to the relevant equivalence class; recording is the ordinary
macroscopic process by which membership in that class becomes a persistent
classical fact.

\subsection{Irreversibility and the arrow of time}

Although each realization of the evolution in the present framework is unitary,
the effective description of measurement is irreversible. This irreversibility
does not arise from a fundamental nonunitary collapse law. It arises from the
combination of stochastic state-space dynamics, the use of equivalence classes,
and the conditioning of subsequent evolution on recorded outcomes.

First, the {\bf (RM)}-driven motion is a stochastic motion in projective
state space of large, and in the idealized case infinite, dimension. A typical
realized path explores new regions of state space, and the probability of
reconstructing its past from its later position is negligible. In finite
dimensions, recurrence may occur in principle, but in very large dimensions the
recurrence times are overwhelmingly long. In infinite-dimensional state space,
the probability of recurrence to a given small neighborhood is effectively
absent. The evolution is local in time and unitary along each realization, but
the stochastic sequence of independently drawn Hamiltonians gives a direction
to typical state-space trajectories. In this sense, an arrow of time is already
present at the level of realized {\bf (RM)} paths.

Second, the ensemble used in {\bf (RM)} is the Gaussian Unitary Ensemble.
The complex random Hamiltonians in this ensemble generate unitary evolution
while scrambling phases and directions in state space. Although each short step
is unitary, the sequence of independently drawn Hamiltonians is not a reversible
record of its own past. Equivalently, the Hamiltonians in the ensemble are not,
in general, invariant under time reversal. Reversing a realized trajectory would
require the corresponding reversed sequence of Hamiltonians, which is not
supplied by the forward stochastic dynamics. Thus the random-matrix dynamics
produces irreversibility at the level of typical paths, and already at the
level of individual steps of the walk, even though each step remains unitary.

Third, equivalence classes introduce a finite-resolution coarse-graining.
A detector does not distinguish all vectors in Hilbert space. It distinguishes
only classes of states that agree within the resolution of the apparatus. In
passing from exact rays to detector-defined equivalence classes, microscopic
information about phases and about distinctions between states inside a class
is no longer part of the operational description. The exact state may continue
to evolve unitarily, but the corresponding description in terms of classical
records has discarded information that is not recoverable from the record.

The role of records was described in the preceding section. Recording does not
cause collapse and does not introduce a new dynamical law. Once the state has
entered a position equivalence class, the subsequent formation of a record is
an ordinary macroscopic process governed by Newtonian dynamics. However, after
such a record is formed, the effective description of later evolution is
conditioned on that record. The recorded class becomes the starting point for
the next stage of the stochastic evolution.

This conditioning gives the operational arrow of time. The past is represented
by a sequence of recorded equivalence classes, while the future remains
probabilistic. Later probabilities are computed relative to the records already
formed; they are not computed from a superposition of all counterfactual
unrecorded alternatives. Repeated environmental monitoring amplifies this
asymmetry, since macroscopic records become correlated with many degrees of
freedom and are stabilized by continual return to classical equivalence-class
sectors.

Thus the arrow of time in the framework has a layered origin. The underlying
dynamics remains unitary for each realization, but typical stochastic
state-space paths are not operationally reversible; equivalence classes discard
microscopic distinctions inaccessible to the detector; and records condition
all subsequent effective evolution. The result is an irreversible classical
history emerging from stochastic yet unitary dynamics in projective state
space.

\subsection{Summary}

The standard paradoxes of quantum mechanics are reformulated in this framework
as questions about the relation between state-space dynamics, classical
submanifolds, detector-defined equivalence classes, and records. Superposition
and entanglement occur when the state evolves away from the relevant classical
submanifold rather than being confined to it. Measurement occurs when the
\({\bf (RM)}\) dynamics drives the state into a detector-defined equivalence
class. The preferred basis is fixed not by an additional postulate, but by the
measurement arrangement, the relevant observable-related submanifold, and the
finite-resolution equivalence relation defined by the apparatus. EPR
correlations arise from the geometry of the joint state space rather than from
signals or physical influences propagating in ordinary space.

Classical motion occurs when frequent environmental interactions keep the state
near the localized phase-space submanifold, so that the tangent component of
Schr\"odinger evolution reproduces Newtonian dynamics. Recording is not a
separate collapse process; once the state has entered a localized equivalence
class, the relevant dynamics is effectively classical, and ordinary macroscopic
dynamics correlates this localized state with a stable detector or
environmental record. Repeated monitoring then stabilizes the record, maintains
Newtonian behavior, and gives the effective irreversibility of the observed
classical history.

Thus the same mechanism---unitary Schr\"odinger evolution supplemented by
{\bf (RM)}-induced stochastic dynamics and finite-resolution equivalence
classes---accounts for state reduction, Born probabilities, classical
trajectories, interference, preferred outcome sectors, macroscopic definiteness,
record formation, Zeno stabilization, irreversibility, and nonlocal
correlations without adding nonunitary collapse or observer-dependent rules.

%%%

\section{Relation to existing approaches}
\label{sec:related-approaches}

The framework developed here overlaps with several existing approaches to the
quantum-to-classical transition, but it differs from them in the role assigned
to state-space geometry and to the effective unitary measurement dynamics
formulated on state space. The dynamics is not formulated primarily on classical
configuration space, nor on a reduced density matrix alone, but on projective
Hilbert space \(\mathbb{CP}^{L_2}\) with its Fubini--Study metric. Classical
configuration space and phase space appear as submanifolds, 
and measurement outcomes are represented by finite-resolution equivalence
classes. The conjecture \({\bf (RM)}\) then supplies an effective isotropic
stochastic dynamics on this same state space. This section compares this
structure with several standard approaches.

\subsection{Decoherence and open-system dynamics}

In standard decoherence theory, a system \(S\) coupled to an environment \(E\)
is described by a joint density operator \(\rho_{SE}(t)\) on the tensor-product
Hilbert space \(\mathcal H_S\otimes \mathcal H_E\). The joint system evolves
unitarily,
\[
\rho_{SE}(t)
=
U(t)\rho_{SE}(0)U(t)^\dagger,
\]
where \(U(t)\) is the unitary evolution operator on
\(\mathcal H_S\otimes \mathcal H_E\). The reduced state of the system alone is
obtained by tracing over the environmental degrees of freedom:
\[
\rho_S(t)=\operatorname{Tr}_E\rho_{SE}(t).
\]
Here \(\operatorname{Tr}_E\) denotes the partial trace over
\(\mathcal H_E\).
For suitable system--environment interactions, the reduced density matrix
becomes approximately diagonal in a preferred pointer basis,
\[
\rho_S(t)
\approx
\sum_\alpha p_\alpha |\alpha\rangle\langle \alpha|,
\]
with off-diagonal terms suppressed. This explains the practical disappearance
of interference between macroscopically distinct alternatives and the stability
of localized or pointer-like states \cite{Zurek2003}.

The present framework is compatible with this mechanism but addresses a
different mathematical problem. Decoherence produces an improper mixture in the
reduced density matrix; it does not by itself specify which individual outcome
occurs in a single run. In the present framework, the state follows a stochastic
unitary path in projective state space,
\[
\varphi(t+dt)
=
\exp\!\left(-\frac{i}{\hbar}\widehat h_{\mathrm{RM}}dt\right)\varphi(t),
\]
and the outcome is the detector-defined equivalence class reached by this path.
Thus the object carrying outcome information in the present framework is not
the reduced density matrix alone, but the state-space path together with the
equivalence-class structure determined by detector resolution.
In particular, decoherence helps explain the stability and robustness of
macroscopic records, while the present framework distinguishes this recording
process from the state-space reduction that selects one equivalence class.

At the level of scales, the two descriptions should agree. For example, in
collisional decoherence the reduced density matrix in the position
representation often has the form
\[
\rho(x,x',t)
=
\rho(x,x',0)
\exp[-\Lambda t (x-x')^2],
\]
in a suitable short-distance approximation. The coefficient \(\Lambda\) is
determined by environmental scattering rates and momentum transfers. In the present framework, the same environmental data enter the \({\bf (RM)}\)
diffusion coefficients. For example, Section~\ref{sec:classicality} estimated
the coefficients
\[
D_p\sim \Gamma q^2,
\qquad
D_a\sim \frac{D_p}{M^2},
\]
where \(\Gamma\) is a collision rate and \(q\) is a typical momentum transfer.
Thus decoherence theory and {\bf (RM)} should agree at the level of
environmental diffusion scales, even though they answer different questions:
decoherence explains suppression of interference in \(\rho_S\), while
{\bf (RM)} supplies a stochastic unitary mechanism for reaching one
detector-defined equivalence class.

\subsection{Caldeira--Leggett-type microscopic bath models}

A representative microscopic model is the Caldeira--Leggett model, in which a
distinguished coordinate \(x\) is linearly coupled to a bath of harmonic
oscillators \cite{CaldeiraLeggett1985}. The Hamiltonian has the schematic form
\[
H
=
\frac{p^2}{2M}+V(x)
+
\sum_j
\left(
\frac{p_j^2}{2m_j}
+
\frac{1}{2}m_j\omega_j^2 q_j^2
\right)
-
x\sum_j c_jq_j
+
x^2\sum_j\frac{c_j^2}{2m_j\omega_j^2}.
\]
After tracing out the bath under suitable approximations, one obtains a master
equation for the reduced density matrix of the system. In the high-temperature
Markovian limit, this has the typical structure
\[
\frac{d\rho}{dt}
=
-\frac{i}{\hbar}[H_S,\rho]
-
\frac{i\gamma}{\hbar}[x,\{p,\rho\}]
-
\frac{D}{\hbar^2}[x,[x,\rho]]
+\cdots ,
\]
where \(\gamma\) is a damping coefficient and \(D\) is a diffusion coefficient.
The double commutator suppresses spatial coherences and gives decoherence in
the position basis.

The role of {\bf (RM)} is different. It is not proposed as a rival microscopic
bath Hamiltonian. Rather, it is an effective conjecture about the coarse-grained
action of many short and complicated environmental interactions on the state
itself. In Caldeira--Leggett-type models, the stochastic or dissipative
structure appears after reduction to the system density matrix. In the present
framework, the effective randomness is placed at the level of the Hamiltonian
generating unitary motion in projective state space.
The comparison is therefore between a microscopic bath model that yields
reduced open-system dynamics and an effective state-space diffusion model. The
two should be compatible when the {\bf (RM)} diffusion coefficients are chosen
to match environmental diffusion scales derived from microscopic models.

This is analogous to the relation between Einstein's theory of Brownian motion,
the related Langevin equation, and microscopic molecular mechanics. Einstein's
assumptions determine the diffusion equation without tracking every molecular
collision, while the Langevin equation gives an effective stochastic
description of the same underlying microscopic process. Similarly, {\bf (RM)}
identifies the structural assumptions -- small independent random unitary steps,
homogeneity, and isotropy in projective state space, needed to derive the Born
rule and the macroscopic stochastic process, without requiring a complete
microscopic derivation of every interaction event.

\subsection{Continuous measurement and quantum trajectories}

Continuous-measurement theory describes monitored quantum systems by stochastic
master equations or stochastic Schr\"odinger equations
\cite{WisemanMilburn,JacobsSteck}. For example, continuous measurement of an
observable \(A\) may be described by a conditioned stochastic equation of the
schematic form
\[
d\rho
=
-\frac{i}{\hbar}[H,\rho]\,dt
-
k[A,[A,\rho]]\,dt
+
\sqrt{2k}
\left(
A\rho+\rho A-2\langle A\rangle\rho
\right)dW_t,
\]
where \(dW_t\) is a Wiener increment and \(k\) is the measurement strength.
Equivalently, one may write a stochastic Schr\"odinger equation for the
conditioned pure state. Such conditioned equations are generally nonlinear in
the normalized state and nonunitary at the level of the conditioned system
evolution. They successfully describe measurement records, feedback, and
quantum trajectories.

The present framework has a similar conditioning structure, but it has a
different geometric and dynamical origin and preserves unitary evolution. In continuous-measurement theory, the measured
observable \(A\) is specified in advance, and the stochastic equation is
constructed to describe conditioning on the corresponding measurement record.
In the present framework, the measured classical variables arise from the
geometry of submanifolds such as \(M_3^\sigma\) and \(M_{3,3}^\sigma\), while
the stochasticity is generated by random Hamiltonians:
\[
d\varphi
=
-\frac{i}{\hbar}\widehat h_{\mathrm{RM}}\varphi\,dt,
\qquad
\widehat h_{\mathrm{RM}}\in \mathrm{GUE}.
\]
For each realization of \(\widehat h_{\mathrm{RM}}\), the evolution remains
unitary. The conditioning occurs only when the recording interaction
assigns the state to a detector-defined equivalence class.

Thus the formal role of \(dW_t\) in continuous measurement is replaced here by
isotropic random unitary motion in projective state space. The outcome sectors
are not merely exact eigenspaces of a preassigned operator; they are
finite-resolution equivalence classes determined by the apparatus and often
associated with localized classical submanifolds. This is why the same
stochastic process gives both classical Brownian errors on \(M_3^\sigma\) and
Born-rule probabilities in \(\mathbb{CP}^{L_2}\), without introducing
fundamental nonlinear or nonunitary dynamics.

\subsection{Ehrenfest theorem and semiclassical limits}

The Ehrenfest theorem \cite{Ehrenfest1927} gives
\[
\frac{d}{dt}\langle \widehat x\rangle
=
\frac{1}{m}\langle \widehat p\rangle,
\qquad
\frac{d}{dt}\langle \widehat p\rangle
=
-\langle \nabla V(\widehat x)\rangle ,
\]
where \(\widehat x\) and \(\widehat p\) are the position and momentum
operators, \(m\) is the particle mass, \(V\) is the potential, and
\(\langle\cdot\rangle\) denotes expectation value in the quantum state.
If the state remains sufficiently localized and \(V\) is sufficiently regular,
then
\[
\langle \nabla V(\widehat x)\rangle
\approx
\nabla V(\langle \widehat x\rangle),
\]
and the expectation values approximately satisfy Newton's equations. This is
an important consistency condition, but it does not explain why localization
persists, why a single outcome occurs, or why the Born rule governs outcome
statistics.

The present construction is different. Newtonian dynamics of a particle is obtained by
restricting the action
\[
S[\varphi]
=
\int
\overline{\varphi}
\left(
i\hbar\partial_t-\widehat h
\right)
\varphi\,d^3x\,dt
\]
to \(M_{3,3}^{\sigma}\). The non-potential terms reduce to the classical form,
and the potential term converges to \(V(a,t)\) under continuity of \(V\).
If \(V\) is differentiable, the restricted action yields
\[
\dot a=\frac{p}{m},
\qquad
\dot p=-\nabla V(a).
\]
Equivalently, the tangent component of the Schr\"odinger vector field on
\(M_{3,3}^{\sigma}\) is the Newtonian vector field. In this restricted
dynamics, the commutator form of the Heisenberg equations reduces to the
corresponding Poisson-bracket form of classical Hamiltonian dynamics.
 The result is therefore a
geometric reduction of the dynamics, not merely an approximation for
expectation values.

The persistence of localization is then supplied by the {\bf (RM)}/environmental
mechanism.This is the role of the drift--diffusion--resolution conditions
\[
\frac{v\,dt}{\sigma}\ll1,
\qquad
\frac{dt}{T_{\mathrm{spr}}}\ll1,
\qquad
D_a\Delta t\ll \sigma^2,
\]
discussed in Section~\ref{sec:classicality}. Thus the Ehrenfest theorem is consistent with the framework, but it is not the
source of the explanation.

\subsection{Objective-collapse models}

Objective-collapse models modify the Schr\"odinger equation by adding
stochastic nonlinear terms. In continuous spontaneous localization (CSL) models, for example, the state
vector satisfies an equation of the general form
\[
d|\psi_t\rangle
=
\left[
-\frac{i}{\hbar}H\,dt
+
\sqrt{\lambda}\int d^3x\,
\bigl(N(x)-\langle N(x)\rangle_t\bigr)dW_t(x)
-
\frac{\lambda}{2}\int d^3x\,
\bigl(N(x)-\langle N(x)\rangle_t\bigr)^2dt
\right]
|\psi_t\rangle ,
\]
where \(\lambda\) is a collapse rate, \(N(x)\) is a smeared mass-density
operator, and \(dW_t(x)\) is a classical noise field
\cite{Bassi2003,Bassi2013}.
The nonlinear terms drive localization and produce single outcomes, but the
dynamics is no longer strictly linear unitary Schr\"odinger evolution.

The present framework also produces single outcomes, but the mechanism is
different. The state evolves by
\[
|\psi_{t+dt}\rangle
=
\exp\!\left[
-\frac{i}{\hbar}
\widehat h_{\mathrm{RM}}dt
\right]
|\psi_t\rangle
\]
during an {\bf (RM)} step. For each realization of \(\widehat h_{\mathrm{RM}}\), the
evolution is unitary. Reduction means that the path in projective
state space enters a detector-defined equivalence class; it is not a
fundamental nonunitary transformation of the state vector.

This distinction changes the empirical interpretation of the parameters.
Collapse models predict effects associated with fundamental noise, such as
spontaneous heating or radiation, and are strongly constrained by
interferometry, heating bounds, and radiation bounds. In the present framework,
there is no universal collapse noise acting independently of measurement or
environmental interaction. The empirical content instead lies in the effective
random-matrix character of complex interactions and in the resulting diffusion
scales. Thus the parameters of {\bf (RM)} are environmental and contextual,
rather than universal collapse constants.

\subsection{Coarse-graining and finite-resolution measurements}

Several approaches emphasize that classical behavior may emerge when
measurements are coarse-grained
\cite{KoflerBrukner2007,JeongLimKim2014,Mukherjee2019,NaikPanigrahi2024}.
In such approaches, limited resolution prevents observation of fine quantum
interference structure, and classical-looking statistics may result.
Mathematically, this often means replacing sharp projectors by coarse-grained
projectors or POVM elements,
\[
P_\Delta
=
\int_\Delta |x\rangle\langle x|\,dx,
\qquad
P(\Delta)=\|P_\Delta\psi\|^2.
\]

The present framework incorporates finite resolution more structurally. A
detector outcome is an equivalence class or finite-resolution sector in
projective state space. For a position region \(\Delta\), the corresponding
probability is
\[
P(\Delta)
=
\int_\Delta |\psi(x)|^2\,dx
=
\cos^2\rho\bigl(\psi,\mathbb P(\mathcal H_\Delta)\bigr),
\]
where
\[
\rho\bigl(\psi,\mathbb P(\mathcal H_\Delta)\bigr)
=
\inf_{\chi\in\mathcal H_\Delta,\ \|\chi\|=1}
\rho(\psi,\chi),
\]
and \(\mathcal H_\Delta\) is the subspace of states supported, to the detector
resolution, in the region \(\Delta\). Thus finite resolution is not only an
observational limitation; it determines the outcome sectors to which the {\bf (RM)}
random walk can return with nonzero probability.

This also resolves a dimensionality issue. In a high- or infinite-dimensional
state space, a random path would generally have zero probability of hitting a
prescribed exact ray. But an actual outcome is not an exact ray. It is an
equivalence class or finite-resolution neighborhood. The relevant event is
entrance into that sector, and its probability is computed by the projection
rule above.

\subsection{Center-of-mass classicalization}

Another route to classical behavior derives the classical motion of the center
of mass of a large quantum system through many-particle concentration or
central-limit-type arguments \cite{Caticha2017,Cui2023}. For a system of \(N\)
particles with center-of-mass coordinate
\[
R=\frac{1}{M}\sum_{i=1}^N m_i x_i,
\]
one typically finds that the center-of-mass distribution becomes sharply peaked
as \(N\) grows. Under broad assumptions, the width of the center-of-mass
distribution decreases like
\[
\Delta R\sim \frac{1}{\sqrt{N}},
\]
or by an analogous mass-weighted concentration estimate. This makes the center
of mass increasingly classical.

This is complementary to, but distinct from, the present approach. In our
construction, classicality is not derived primarily from the large-\(N\)
concentration of a configuration-space distribution. Instead, the main object is
the localized phase-space submanifold
and Newtonian dynamics is the tangent component of Schr\"odinger evolution on
this submanifold. Gaussian representatives are convenient and physically
natural, especially for macroscopic center-of-mass states, but the geometric
derivation of the classical action does not require exact Gaussianity. It
extends to equivalence classes of sufficiently localized states.

The central-limit picture nevertheless supports the physical plausibility of
Gaussian representatives in macroscopic regimes. If a macroscopic center of mass
is influenced by many weakly correlated microscopic degrees of freedom, then
smooth approximately Gaussian localization is expected. 
Thus center-of-mass concentration helps explain why Gaussian representatives
are typical, while the present framework explains why Gaussianity is not
structurally essential for the transition to classicality or for measurement
dynamics.

\subsection{Bohmian mechanics}

Bohmian mechanics \cite{Bohm1952I,Bohm1952II} supplements the wave function
\[
\Psi=\Psi(q_1,\ldots,q_N,t)
\]
with actual particle positions \(Q_1(t),\ldots,Q_N(t)\). These positions are
guided by
\[
\dot Q_k(t)
=
\frac{\hbar}{m_k}
\operatorname{Im}
\left.
\frac{\nabla_k\Psi(q_1,\ldots,q_N,t)}
{\Psi(q_1,\ldots,q_N,t)}
\right|_{q_j=Q_j(t)} .
\]
This gives definite trajectories in configuration space, but requires
additional hidden variables and a guiding equation that is nonlocal on
configuration space for entangled states.

The present framework introduces no additional particle positions. The
fundamental object is the state as a point in projective Hilbert space.
Classical positions appear when the state lies in, or is recorded as belonging
to, localized equivalence classes associated with a classical space submanifold.
Thus the classical trajectory is not a hidden path \(Q(t)\) in
\(\mathbb R^3\), but a sequence of records produced by returns of the
state-space path to the classical sector.

\subsection{Everettian approaches}

Everettian approaches \cite{Everett1957,DeWittGraham1973} preserve unitary evolution by interpreting measurement as
branching into multiple realized alternatives,
\[
\sum_\alpha c_\alpha |\alpha\rangle |A_0\rangle
\longrightarrow
\sum_\alpha c_\alpha |\alpha\rangle |A_\alpha\rangle .
\]
The present framework also preserves unitary evolution for each realization of
the dynamics, but it does not interpret all terms in the superposition as
simultaneously realized classical worlds. The \({\bf (RM)}\) Hamiltonian
selects a stochastic state-space path, and an observed outcome is the
detector-defined equivalence class reached by that path. Thus outcome
definiteness is obtained through stochastic unitary dynamics and conditioning,
rather than through branching.

\subsection{Summary}

The comparison can be summarized as follows. Decoherence explains the
suppression of off-diagonal terms in reduced density matrices.
Caldeira--Leggett-type models derive open-system coefficients from microscopic
bath Hamiltonians. Continuous-measurement theory gives stochastic conditioned
equations once the measured observable is specified. Collapse models add
fundamental nonunitary noise. Coarse-graining approaches emphasize finite
resolution, and center-of-mass classicalization explains the concentration of
large systems. Bohmian mechanics introduces additional particle positions guided
by the wave function, while Everettian approaches interpret measurement in
terms of branching into multiple realized alternatives.

The present framework combines different geometric and dynamical ingredients.
Classical space and phase space are submanifolds of projective Hilbert space.
The Schr\"odinger action restricted to the phase-space submanifold becomes the
classical action. Measurement outcomes are finite-resolution equivalence
classes. Finally, \({\bf (RM)}\) supplies isotropic stochastic unitary dynamics
on projective state space. This combination yields, within one framework,
normal distributions for classical measurement errors, Born probabilities for
microscopic measurements, and stroboscopic Newtonian motion for macroscopic
bodies.

\section{Conclusion}
\label{sec:conclusion}

We have proposed a unified geometric and stochastic framework for understanding
measurement, state reduction, and the quantum-to-classical transition. The
geometric part of the framework identifies classical configuration space and
classical phase space with localized submanifolds of projective quantum state
space. The additional dynamical ingredient is the conjecture {\bf (RM)}:
complex interactions with measuring devices and environments are effectively
described, after coarse-graining, by random-matrix Hamiltonians generating
isotropic stochastic unitary motion on projective state space.

The geometric construction explains how classical mechanics is contained in
quantum mechanics. The manifolds \(M_3^\sigma\) and \(M_{3,3}^\sigma\) carry
the Euclidean geometry of the classical configuration space and phase space of
a particle through the metric induced by the Fubini--Study metric. When the
Schr\"odinger action is restricted to \(M_{3,3}^\sigma\), it becomes the
classical action, and the tangent component of the Schr\"odinger vector field
reproduces Newtonian dynamics. Thus Newtonian motion appears not as an
independent postulate, but as the tangent dynamics of quantum states constrained
to the classical phase-space submanifold.

The conjecture {\bf (RM)} explains how measurement and state reduction can
arise within unitary dynamics. It produces homogeneous and isotropic stochastic
motion in projective state space. When restricted to the classical submanifold
\(M_3^\sigma\simeq \mathbb R^3\), this motion becomes ordinary Brownian-type
diffusion of a particle and yields the normal distribution of classical
measurement errors. In the full projective state space, the same isotropy makes
transition probabilities depend only on Fubini--Study distance. Since the
distance-dependent transition law is fixed on \(M_3^\sigma\), where it agrees
with the Born probability for localized states, it extends uniquely to the Born
rule on the full projective state space. The construction extends to arbitrary
systems of particles.

%%%

%%%

Equivalence classes of detector-indistinguishable states play an essential
role. A measurement outcome is not an exact ray in Hilbert space, but a
finite-resolution outcome sector determined by the measuring device. The
probability of an outcome is the probability that, at the recording time, the
{\bf (RM)}-driven state is assigned to the corresponding sector. This avoids
the difficulty of requiring a random path in an infinite-dimensional space to
hit an exact state or exact lower-dimensional submanifold. It also clarifies
why wave-function tails do not prevent single outcomes: what is recorded is
membership in an equivalence class, not pointwise disappearance of the wave
function elsewhere.

The same mechanism yields different physical regimes. For microscopic
measurements, the {\bf (RM)} interaction can dominate over the free
Hamiltonian during the short measurement interval. The state then undergoes
essentially isotropic diffusion in projective state space during the
measurement interval and is recorded in one of the detector-defined outcome
classes, with Born probabilities. For macroscopic bodies, the free
Schr\"odinger flow supplies Newtonian tangent drift, while frequent
environmental {\bf (RM)} interactions continually return the state to a
narrow neighborhood of \(M_{3,3}^\sigma\). Under the
drift--diffusion--resolution conditions derived above, the resulting recorded
positions form a stroboscopic Newtonian trajectory.

The framework therefore gives a single account of several phenomena that are
often treated separately: classical space geometry, Newtonian dynamics, normal
distributions in classical measurement, Born probabilities in quantum
measurement, single outcomes, and the stability of macroscopic classical
behavior. It also gives a common interpretation of standard experiments and
paradoxes. Interference and entanglement occur when the state evolves outside
the relevant classical submanifold. Measurement occurs when the recording
interaction assigns the state to a detector-defined outcome class. Macroscopic
definiteness results from repeated environmental stabilization of classical
sectors. EPR correlations arise from the geometry of joint state space rather
than from signals propagating in classical space.

%%%
The status of {\bf (RM)} is stronger than that of an arbitrary
random-matrix ansatz. Classical measurement errors can be modeled by Brownian
motion arising from a Newtonian bath. At the same time, Newtonian dynamics is
the restriction of Schr\"odinger dynamics to the classical phase-space
submanifold. It follows that the Brownian bath dynamics admits a unitary lift
to projective Hilbert space. According to
Theorem~\ref{thm:Brownian-lift-RM}, in the Gaussian random-walk approximation,
translation invariance of the total step distribution on \(M_3^\sigma\) fixes
the corresponding unitary lift, which is generated by the GUE random
Hamiltonians of {\bf (RM)}. In this sense, {\bf (RM)} is the natural unitary lift of classical Brownian
measurement dynamics, and is unique under the assumptions of
Theorem~\ref{thm:Brownian-lift-RM}.

%%%

This result is also supported by the widespread role of random Hamiltonians and
random matrices in the description of fluctuations in complex quantum systems.
Moreover, the choice of the GUE ensemble is imposed upon us in this framework.
The Gaussian Orthogonal Ensemble would not be sufficient for the present
purpose: GOE is invariant only under real orthogonal transformations and
therefore selects a preferred real structure in Hilbert space. It does not give
the full unitary-invariant distribution of random Hamiltonian directions needed
for isotropic diffusion in complex projective Hilbert space with the
Fubini--Study geometry, as required in the framework.

It is not claimed here that the random-matrix form of the interaction Hamiltonian has already been derived from a complete microscopic model of every measuring apparatus or environment. Rather, the point is that the remaining physical assumption is narrower: complex environmental and measurement interactions must realize, after coarse-graining, the unitary lift of classical Brownian measurement dynamics described by Theorem~\ref{thm:Brownian-lift-RM}. Under this assumption, the GUE form of \({\bf (RM)}\) is fixed rather than chosen independently. The resulting consequences are mathematically precise and physically nontrivial. They include unitary evolution for each realization, single Born-distributed outcomes, and recovery of Newtonian motion for macroscopic systems.
%%%

Several directions remain open. First, one should derive the {\bf (RM)}
form from more microscopic models of environmental interaction, or at least
identify broad dynamical conditions under which GUE statistics emerge. Second,
the parameter estimates should be refined and tested for concrete experimental
systems, including mesoscopic interferometers, optomechanical devices, and
continuous-measurement setups. Third, the extension to relativistic fields and
quantum field theory should be developed systematically, with classical field
configurations and source dynamics represented as appropriate submanifolds of
field-state space. Finally, possible deviations from ideal {\bf (RM)}
behavior may provide additional experimental tests of the proposed mechanism.

The main conclusion is that the classical world need not be added to quantum
mechanics from outside. Classical space, classical phase space, and Newtonian
dynamics arise inside projective quantum state space, while measurement
outcomes and Born probabilities arise from stochastic unitary motion between
detector-defined equivalence classes. If the random-matrix conjecture {\bf (RM)} is
physically correct, then quantum mechanics contains within its own geometry and
unitary dynamics the mechanisms needed for state reduction, classicality, and
the quantum-to-classical transition.

%Numerical simulations of the {\bf (RM)}-induced diffusion, its Brownian restriction to the classical submanifold, and the resulting Born-rule and macroscopic classical regimes will be presented in a separate companion work.

%%%

\section*{Declaration of interest statement}

The author declares no competing interests.

\end{document}